\documentclass[11pt]{article}

\usepackage[letterpaper,margin=1in]{geometry}
\usepackage{amsmath,amssymb,amsthm,mathtools}
\usepackage{algorithm}
\usepackage[noend]{algpseudocode}
\usepackage{booktabs}
\usepackage{tabularx}
\usepackage{enumitem}
\usepackage{microtype}
\usepackage{aliascnt}
\usepackage{hyperref}
\usepackage[nameinlink,capitalize,noabbrev]{cleveref}

\hypersetup{
 colorlinks=true,
 linkcolor=blue,
 citecolor=blue,
 urlcolor=blue,
 pdftitle={Random-Order Online Facility Location Beyond Uniform Opening Costs},
 pdfauthor={Bo Peng and Zhihao Gavin Tang}
}
\allowdisplaybreaks

\makeatletter
\providecommand*{\theHALG@line}{\thealgorithm.\arabic{ALG@line}}
\makeatother

\newtheorem{theorem}{Theorem}[section]

\newaliascnt{lemma}{theorem}
\newtheorem{lemma}[lemma]{Lemma}
\aliascntresetthe{lemma}

\newaliascnt{proposition}{theorem}
\newtheorem{proposition}[proposition]{Proposition}
\aliascntresetthe{proposition}

\newaliascnt{corollary}{theorem}
\newtheorem{corollary}[corollary]{Corollary}
\aliascntresetthe{corollary}

\newaliascnt{claim}{theorem}

\aliascntresetthe{claim}

\theoremstyle{definition}
\newaliascnt{definition}{theorem}
\newtheorem{definition}[definition]{Definition}
\aliascntresetthe{definition}

\newaliascnt{remark}{theorem}
\newtheorem{remark}[remark]{Remark}
\aliascntresetthe{remark}

\crefname{theorem}{Theorem}{Theorems}
\crefname{lemma}{Lemma}{Lemmas}
\crefname{proposition}{Proposition}{Propositions}
\crefname{corollary}{Corollary}{Corollaries}
\crefname{claim}{Claim}{Claims}
\crefname{definition}{Definition}{Definitions}
\crefname{remark}{Remark}{Remarks}

\newcommand{\F}{\mathcal F}

\newcommand{\OPT}{\mathrm{OPT}}
\newcommand{\ALG}{\mathrm{ALG}}
\newcommand{\E}{\mathbb E}
\newcommand{\Prb}{\mathbb P}
\newcommand{\1}{\mathbf 1}
\newcommand{\R}{\mathbb R}

\newcommand{\eps}{\varepsilon}

\title{Random-Order Online Facility Location\\
Beyond Uniform Opening Costs}
\author{
Bo Peng\thanks{
Email: \texttt{ahqspbo@gmail.com}
}
\qquad
Zhihao Gavin Tang\thanks{
Email: \texttt{tang.zhihao@mail.shufe.edu.cn}
}\\[1mm]
Shanghai University of Finance and Economics
}
\date{}

\begin{document}
\hypersetup{pageanchor=false}
\pagenumbering{gobble}

\begin{titlepage}
\maketitle
\vspace{-1.5em}

\begin{abstract}
We study online metric facility location in the random-order model with
arbitrary positive opening costs. A finite set of candidate facilities and
their costs is known in advance, while an adversary fixes a multiset of demand
points that arrives in a uniformly random order. This setting includes both
prescribed candidate sites and the classical finite full-space node-cost
model.

For a known horizon, we give a deterministic $4.2674$-competitive algorithm,
improving the previous factor $33$ for nonuniform opening costs.  At rank $t$, the algorithm uses the
positive normalized rank $q_t=t/n$, chooses a candidate minimizing
$d(x,y)+\lambda_t f_y$, where
$\lambda_t=\min\{1,q_t/\mu\}$, and opens it when the current connection
distance covers this penalized objective.  The analysis uses a monotone
one-round charge and an upper-envelope decomposition to control later points
and the first point of each optimal cluster.  With unit opening costs, the
rule reduces exactly to a cutoff on the distance improvement attainable from
a nearest candidate.  A supplementary appendix gives the sharper analysis
of the closely related zero-start rank cutoff and obtains a ratio below
$3.2805$.

We also prove a $3-o(1)$ lower bound for arbitrary randomized online
algorithms. The lower bound already holds with uniform costs on a prescribed
candidate set and transfers, without loss, to the finite full-space model with
nonuniform opening costs. Together with the recent competitive ratio below $2.42$ for full-space
uniform costs, this yields a strict separation between the full-space
uniform- and nonuniform-cost models.
\end{abstract}
\thispagestyle{empty}
\end{titlepage}

\tableofcontents
\clearpage
\hypersetup{pageanchor=true}
\pagenumbering{arabic}

\section{Introduction}

In online metric facility location, demand points arrive one by one.  When a
demand point arrives, an online algorithm may irrevocably open facilities and
must immediately assign the point to an open facility.  The objective is to
minimize the sum of facility-opening costs and assignment distances.  The
problem was introduced by Meyerson~\cite{Meyerson2001} and has become a
canonical example in online clustering and random-order analysis.  Subsequent
work has studied adversarial and random arrival orders
\cite{Fotakis2008,KaplanNaoriRaz2023,HuangJiang2026}, prescribed facility
sites \cite{AnagnostopoulosEtAl2004,LiMiaoWuXu2026}, recourse and deletions
\cite{GuoEtAl2020,CyganEtAl2018}, and algorithms augmented with predictions,
advice, or samples
\cite{AlmanzaEtAl2021,JiangEtAl2022,FotakisEtAl2025,ArgueEtAl2022}.

The sharpest random-order results have focused on the classical full-space
uniform-cost model, in which every metric point can be opened at the same
cost.  Meyerson's natural distance-proportional rule was originally shown to
be $8$-competitive.  Kaplan, Naori, and Raz~\cite{KaplanNaoriRaz2023} proved
that this rule is exactly $4$-competitive, obtained a tight ratio of $3$ for
its best linear scaling, and proved a lower bound of $2$ for every online
algorithm.  Huang and Jiang~\cite{HuangJiang2026} subsequently incorporated
the normalized arrival rank into the opening decision and obtained a
deterministic competitive ratio below $2.42$.

We study the problem in the explicit \emph{candidate-site formulation}.  The
algorithm is given a finite set $\F\subseteq M$, may open facilities only in
$\F$, and each $y\in\F$ has an arbitrary positive opening cost $f_y$.  On a
finite metric, this is equivalent at the level of facility feasibility to the
classical node-cost formulation: assigning cost $+\infty$ outside $\F$
forbids those points, while the finite-cost nodes form the candidate set.
Thus $\F=M$ recovers the finite full-space model, and restricting the costs on
$\F$ to a common value recovers the prescribed-site uniform-cost model.

For nonuniform opening costs, Meyerson~\cite{Meyerson2001} gave a
$33$-competitive random-order algorithm in the node-cost formulation.  By the
equivalence above, that guarantee also covers restricted candidate sets,
although it was not stated using an explicit candidate set.  More recently,
Li, Miao, Wu, and Xu~\cite{LiMiaoWuXu2026} studied the prescribed-site
uniform-cost restriction explicitly and obtained an $8$-competitive
algorithm.  Our formulation simultaneously contains these two settings:
Meyerson's finite nonuniform model is obtained by taking $\F$ to be the
finite-cost nodes, whereas the model of Li et al.\ is obtained by assigning
one common cost to all sites in $\F$.

Our main upper-bound result is a deterministic competitive ratio below
$4.2674$.  This improves Meyerson's factor $33$ for nonuniform opening costs
and, as a special case, the factor $8$ of Li et
al.~\cite{LiMiaoWuXu2026}.  The algorithm uses a rank-dependent
penalized-distance comparison.  At round $t$, it selects a candidate
minimizing its connection distance plus an opening-cost penalty scaled by the
normalized rank $t/n$, and opens the selected candidate when the resulting
penalized distance does not exceed the current connection distance.  The
analysis combines a monotone one-round charge with an upper-envelope
decomposition that controls both later requests in an optimal cluster and the
excess cost of the facility selected by the first request.

The normalized rank $t/n$ requires the horizon $n$ to be known in advance.
This is an additional information assumption relative to Meyerson's original
algorithm and the $8$-competitive algorithm of Li et al., neither of which
uses the eventual number of requests.  Apart from the horizon, the online
information pattern is unchanged: the metric, candidate sites, and opening
costs are known initially, whereas the demand multiset is hidden and revealed
in uniformly random order.  We do not claim the same $4.2674$ guarantee when
the horizon is unknown.

Under unit opening costs, the penalized-distance choice reduces to selecting a
nearest candidate, and its opening condition becomes a cutoff on the
attainable reduction in connection distance.  Appendix
\ref{app:uniform-candidates} studies this specialization using a zero-start
rank convention.  By adapting the rank-cutoff framework of Huang and Jiang
and isolating the effect of prescribed candidate locations, it gives the
supplementary bound below $3.2805$.

Our second main result is a $3-o(1)$ lower bound for arbitrary randomized
online algorithms.  The lower bound already holds in the prescribed-site
model with uniform costs, even when the horizon is revealed to the algorithm.
A replacement reduction transfers it to the finite full-space model with
nonuniform opening costs.  Together with the sub-$2.42$ full-space
uniform-cost upper bound of Huang and Jiang~\cite{HuangJiang2026}, this yields
a strict separation between the optimal competitive ratios of the full-space
uniform- and nonuniform-cost models.

Appendix~\ref{app:one-center} further shows that the value three is tight for
a broad single-facility benchmark.  Consequently, any lower bound strictly
above three must exploit an optimal solution with a genuinely hidden
multi-facility structure.

\paragraph{Results at a glance.}
The table summarizes the finite random-order models most relevant here.
The prescribed-site lower bound of $3$ already holds with uniform opening
costs, and its full-space nonuniform counterpart follows from the
all-finite-cost transfer proved below.  Since the two nonuniform-cost
formulations consequently have the same entries, we combine them in one row.
The $3.2805$ entry is the sharpened unit-cost analysis in
Appendix~\ref{app:uniform-candidates}, whereas the main upper-bound theorem is
the $4.2674$ guarantee for arbitrary positive opening costs.  The upper-bound results established in this paper assume that the algorithm
is given the horizon $n$ in advance.  The lower bounds proved here remain
valid even when this information is available to the algorithm.

\begin{center}
\small
\begin{tabularx}{\textwidth}{@{}lXX@{}}
\toprule
Model & Best upper bound & Lower bound \\
\midrule
Full space, uniform costs
  & $<2.42$~\cite{HuangJiang2026}
  & $2$~\cite{KaplanNaoriRaz2023} \\
Prescribed sites, uniform costs
  & $<3.2805$ (Appendix~\ref{app:uniform-candidates});
    previously $8$~\cite{LiMiaoWuXu2026}
  & $3$ (this paper) \\
\shortstack[l]{Full space or prescribed sites,\\nonuniform costs}
  & $<4.2674$ (this paper);
    previously $33$~\cite{Meyerson2001}
  & $3$ (this paper) \\
\bottomrule
\end{tabularx}
\end{center}

The main algorithmic issue is how one arrival should compare candidates with
very different opening costs.  We address it directly through a Lagrangian
objective: early requests discount opening cost more strongly, and the
algorithm chooses the candidate with minimum penalized distance.  This avoids
separate decisions at multiple cost scales and always opens at most one
facility.  The proof must nevertheless control how the selected candidate
changes with time; the upper concave envelope of its distance improvement as
a function of opening cost supplies exactly this structure.

\subsection{Main results and technical overview}

\paragraph{A deterministic upper bound for arbitrary opening costs.}
We introduce \emph{Penalized-Distance RankCut}, a deterministic algorithm for
arbitrary positive opening costs on a prescribed candidate set.  Suppose that
the current demand point is $x$ and that it appears in position $t$ among the
$n$ requests.  For every candidate facility $y\in\F$, let $f_y$ be its opening
cost and form the penalized distance
\[
 d(x,y)+\lambda_t f_y,
 \qquad
 \lambda_t:=
 \min\left\{1,\frac{t}{\mu n}\right\},
\]
where $\mu\in(0,1]$ is a fixed parameter.  The algorithm selects a candidate
minimizing this quantity, using fixed deterministic tie-breaking.  It opens
the selected candidate if the distance from $x$ to the nearest currently open
facility is at least the selected penalized distance; otherwise it opens no
facility.  Thus each arrival compares all candidates in a single optimization
and opens at most one facility, without discretizing the opening costs or
making separate decisions at different cost scales.

The multiplier $\lambda_t$ makes the algorithm more willing to invest in a
facility early in the permutation, when it may serve many subsequent
requests, and more conservative later.  Optimizing this tradeoff yields a
competitive ratio below $4.2674$.  The algorithm has no internal randomness;
the expectation in the guarantee is taken only over the uniformly random
input order.

\paragraph{Structure of the upper-bound analysis.}
We fix an optimal solution, partition the requests into its facility clusters,
and use the first arrival in each cluster as an anchor.  The principal
difficulty is that the candidate selected by a later request and its current
connection distance depend on the entire intervening history.

The first ingredient is a monotone one-round upper bound.  For a fixed request
and rank, the algorithm's cost is dominated by a two-piece function of the
current connection distance.  This function changes from connection cost to
opening-and-connection cost at the algorithm's cutoff, with a nonnegative jump
at the transition.  Once the anchor has been served, a facility serving it
remains open.  Monotonicity and the triangle inequality then allow the
history-dependent connection distance of each later request to be replaced by
a bound determined by the anchor.

The second ingredient is an upper-envelope decomposition of the candidate
choices.  Each candidate is represented by its opening cost and the connection
improvement it can provide.  The penalized-distance minimization selects a point on the resulting
upper envelope. Lemma~\ref{lem:envelope} decomposes the selected
point into consecutive coordinate increments and shows that an
increment can be included only when its slope is at least the current
penalty parameter. Random-order rank estimates therefore charge its
expected cost to the corresponding connection improvement.  The same
decomposition controls the expensive part of the facility opened at the
anchor, while a separate first-rank certificate accounts for its base cost and
residual connection distance.

Together, these arguments bound the expected online cost of each optimal
cluster by a linear combination of its optimal opening and connection costs.
Balancing the two coefficients through the parameter $\mu$ and summing over
the clusters gives the stated upper bound.

\paragraph{The unit-cost specialization.}
When all opening costs are equal, the penalized term is common to every
candidate.  The algorithm therefore selects a nearest candidate, and its
opening test reduces to a cutoff on the connection-distance improvement
attainable from that candidate.  Appendix~\ref{app:uniform-candidates}
analyzes the closely related zero-start rank convention and, by adapting the
rank-cutoff framework of Huang and Jiang~\cite{HuangJiang2026}, obtains the
sharper supplementary bound below $3.2805$.

The positive rank used in the main theorem and the zero-start rank used in
the appendix are, respectively, the right- and left-endpoint discretizations
of the same continuous time scale.  They define distinct deterministic
clocks and are analyzed separately in their respective settings; see
\cref{rem:continuous-clock}.

\paragraph{The lower bound and full-space separation.}
The lower bound is first established in the prescribed-candidate model with
uniform opening costs.  The construction has a large universe of possible
request locations and a candidate facility associated with every support of a
prescribed size.  The offline optimum can select the facility tailored to the
realized support, whereas an online algorithm must make its opening decisions
before that support is known.

A pointwise accounting identity relates the online cost to the facilities it
opens and the requests covered by them.  A coverage-credit argument shows that
the expected benefit of additional openings cannot offset their cost.
Conditioning on the realized request multiset then produces a fixed
random-order instance on which every randomized online algorithm has
competitive ratio $3-o(1)$.

Finally, we transfer this construction to the classical full-space
nonuniform-cost model.  Every metric point is made feasible and assigned a
finite opening cost, and any opening at a demand location can be replaced
online by an original candidate facility without increasing the total cost.
The factor-three lower bound therefore persists even when no location is
forbidden.  Together with the full-space uniform-cost result of Huang and
Jiang, this gives
\[
 \operatorname{CR}^{\mathrm{full}}_{\mathrm{unif}}
 <2.42<3
 \le
 \operatorname{CR}^{\mathrm{full}}_{\mathrm{nonunif}},
\]
and hence a strict separation between the uniform- and nonuniform-cost
full-space models.

\subsection{Related work}\label{sec:related}

\paragraph{Facility location under adversarial and random order.}
Meyerson~\cite{Meyerson2001} introduced online metric facility location and
studied both adversarial and random arrival orders. For adversarial order,
Fotakis established the asymptotically optimal
$\Theta(\log n/\log\log n)$ competitive ratio~\cite{Fotakis2008} and later
developed a primal--dual algorithm for nonuniform opening
costs~\cite{Fotakis2007}. \mbox{Anagnostopoulos} et
al.~\cite{AnagnostopoulosEtAl2004} gave a simple deterministic
$O(\log n)$-competitive algorithm that applies to several
facility-feasibility models. Fotakis' survey~\cite{FotakisSurvey2011}
provides a broader account of online and incremental facility location.

In the random-order model, Kaplan, Naori, and
Raz~\cite{KaplanNaoriRaz2023} proved that Meyerson's uniform-cost
DistProb rule is exactly $4$-competitive, identified the best scaling in the
$q$-DistProb family with tight ratio $3$, and established the classical lower
bound of $2$ for arbitrary algorithms. Huang and
Jiang~\cite{HuangJiang2026} subsequently moved beyond time-oblivious
DistProb rules and obtained a deterministic ratio below $2.42$ by using the
arrival time in the opening decision.  They also proved their cutoff rule
optimal among TimeDist algorithms by balancing a sparse-star instance against
repeated dense locations.  Appendix~\ref{app:uniform-candidates} applies the
same single-scale architecture to the prescribed-candidate marginal signal
$(d(x,O)-d(x,\F))^+$.  This signal is also exactly what the main
penalized-distance rule becomes at unit costs.  The appendix uses the same
zero-start rank normalization as Huang and Jiang for the direct comparison:
the common cutoff and class-lower-bound steps are parallel to theirs, while the
candidate-specific difference is the residual connection distance after
opening a nearest feasible site.
Lang~\cite{Lang2018} analyzed
Meyerson's rule against a $t$-bounded adversary. For uniform costs on an explicit prescribed set of potential facility
locations, Li, Miao, Wu, and Xu~\cite{LiMiaoWuXu2026} obtained an
$8$-competitive horizon-independent algorithm.  Their feasibility model is
the unit-cost special case of ours, whereas our rank-based upper bound allows
nonuniform costs but assumes that the horizon is known.

\paragraph{The random-order framework and robustness of arrival assumptions.}
Random order is a central beyond-worst-case model for online algorithms; see
the survey of Gupta and Singla~\cite{GuptaSingla2021}. Representative results
show substantial improvements for online covering, packing, matching,
knapsack, and load balancing
\cite{GuptaKehneLevin2022,KesselheimEtAl2018,KaplanNaoriRaz2022,
AlbersKhanLadewig2021,Molinaro2017}. Several works study guarantees that
interpolate between a uniformly random permutation and more adversarial
arrival processes, including nonuniform arrival distributions, robust
secretary models, bursty adversaries, and algorithms that perform well in
both stochastic and adversarial regimes
\cite{KesselheimKleinbergNiazadeh2015,BradacEtAl2020,
KesselheimMolinaro2020,MirrokniEtAl2012}. These works motivate asking which
parts of a random-order proof genuinely use the uniform permutation.  Our
main algorithm is deterministic and uses the positive normalized rank
$q_t=t/n$.  The proof uses only two permutation facts: conditional on the
first cluster rank, each fixed later copy is uniform among the remaining
ranks, and a union bound controls how early the first cluster rank can be.
The one-rank shift from Huang and Jiang's $z_t=(t-1)/n$ is essential with
nonuniform costs by \cref{prop:zero-start-unbounded}.  At unit costs the
obstruction disappears, and Appendix~\ref{app:uniform-candidates} uses the
zero-start clock to stay in their rank-based comparison class.

\paragraph{Streaming, recourse, deletions, and leasing.}
Facility-location ideas have also played an important role in streaming
clustering, where online openings are used to maintain compact summaries of
large data streams~\cite{CharikarEtAl2003,GuhaEtAl2003}. A complementary
line of work relaxes irrevocability or changes the temporal model. Guo,
Kulkarni, Li, and Xian~\cite{GuoEtAl2020} studied online and dynamic facility
location with recourse. Cygan et al.~\cite{CyganEtAl2018} considered facility
location with deletions. Nagarajan and
Williamson~\cite{NagarajanWilliamson2013} studied facility leasing, which
generalizes online facility location by allowing several lease durations.
These models obtain additional flexibility through recourse, departures, or
time-limited openings, whereas our facilities and assignments remain
irrevocable.

\paragraph{Predictions, advice, and samples.}
Another line of work supplements the online input with auxiliary information.
Almanza et al.~\cite{AlmanzaEtAl2021} considered multiple advice sources for
online facility location. Jiang et al.~\cite{JiangEtAl2022} studied
predictions of the optimal serving facility, and Fotakis et
al.~\cite{FotakisEtAl2025} obtained improved prediction-error guarantees.
Argue et al.~\cite{ArgueEtAl2022} gave sample-assisted algorithms that include
facility-location clustering, building on the broader sample-based online
framework~\cite{KaplanNaoriRaz2020}. Learning-augmented online graph
algorithms provide a related general perspective on consistency and
robustness~\cite{AzarPanigrahiTouitou2022}. These information-augmented
models are orthogonal to the present work: our algorithm receives no sample,
prediction, advice, or recourse and exploits only the uniformly random arrival
order.

\subsection{Organization}

The remainder of the paper is organized as follows.
\Cref{sec:model} defines the explicit candidate-site model and records its
precise relation to Meyerson's node-cost formulation.  \Cref{sec:algorithm}
presents the positive normalized-rank clock and Penalized-Distance RankCut.
\Cref{sec:analysis} proves the $4.2674$ upper bound using monotonicity,
upper-envelope decompositions, and a cluster-based analysis.
\Cref{sec:lower} proves the restricted-candidate lower bound, transfers it to
the full-space nonuniform node-cost model, and derives the uniform/nonuniform
separation.  \Cref{sec:discussion} discusses the implications of these results
and directions for improving the upper bound.
Appendix~\ref{app:uniform-candidates} gives the sharper zero-start unit-cost
analysis and its comparison with the full-space TimeDist framework.
Appendix~\ref{app:one-center} explains why lower bounds based on a single
hidden optimal center cannot exceed three.

\section{Problem definition}\label{sec:model}

Let $(M,d)$ be a metric space, and let $\F\subseteq M$ be a finite, nonempty
set of candidate facilities.  Each $y\in\F$ has a positive opening cost
$f_y>0$.  An adversary fixes a multiset $U$ of $n$ demand-point copies in $M$.
For the purposes of permutation ranks, copies located at the same metric point
are treated as distinct labeled objects.  These labels are used only in the
analysis and are not revealed to the algorithm.  The $n$ copies arrive in a
uniformly random order.  The algorithm knows $(M,d)$, $\F$, all opening costs,
and the horizon $n$, but not the multiset $U$.

Let $O_t\subseteq\F$ be the set of facilities open after round $t$.  When demand point $v_t$ arrives, the algorithm may open facilities and then
irrevocably assigns $v_t$ to an open facility.  Its total cost is
\[
 \ALG
 =
 \sum_{y\in O_n}f_y+\sum_{t=1}^n d(v_t,O_t),
 \qquad
 d(x,O):=\min_{y\in O}d(x,y),
\]
where $d(x,\varnothing)=+\infty$.

The offline optimum sees the entire multiset and chooses a set of facilities
afterward:
\[
 \OPT
 =
 \min_{O\subseteq\F}
 \left\{
 \sum_{y\in O}f_y+\sum_{u\in U}d(u,O)
 \right\}.
\]
An online algorithm is $c$-competitive in the random-order model if, for every
fixed multiset $U$,
\[
 \E[\ALG]\le c\,\OPT,
\]
where the expectation is over the uniformly random arrival order and the
algorithm's internal randomness.

We use an explicit candidate set because it cleanly separates facility
feasibility from opening cost.  On a finite metric, this is equivalent to
Meyerson's node-cost formulation: assigning cost $+\infty$ to every point in
$M\setminus\F$ forbids those points, while the finite-cost nodes form the
candidate set.  Thus taking $\F=M$ recovers the classical full-space model,
and taking $\F$ to be the finite-cost nodes recovers a finite nonuniform
node-cost instance.  All costs manipulated by our algorithm are the positive,
finite costs on $\F$.

% \begin{remark}[Relation to the classical node-cost formulation]
% Meyerson~\cite{Meyerson2001} did not state a separate candidate set
% $\F\subseteq M$; instead, every metric node carries an opening cost.  Because
% his nonuniform formulation allows infinite-cost nodes, it implicitly includes
% restricted candidate-site instances.  Our upper bound is therefore an
% improvement for that general node-cost formulation, expressed in the explicit
% candidate-site notation used by our penalized-distance rule.
% \end{remark}
\begin{remark}[Information assumptions]
Our upper-bound algorithm does make an additional information assumption: the
horizon $n$ is known in advance.  This information is used to form the
normalized rank $q_t=t/n$.  Meyerson's original rule and the
$8$-competitive rule of Li et al. do not require this information.  We do not
claim the same $4.2674$ guarantee for an unknown horizon.  The lower bounds,
on the other hand, continue to hold even when the algorithm is given the
horizon.
\end{remark}

\section{The algorithm}\label{sec:algorithm}

Our algorithm uses the normalized arrival rank as a multiplier on opening
cost.  It compares all candidate facilities in one optimization, rather than
rounding their costs or separating them into levels.  The horizon $n$ is
known, and at round $t$ we use the positive deterministic clock
\begin{equation}
 q_t:=\frac{t}{n},
 \qquad
 \lambda_t:=\min\left\{1,\frac{q_t}{\mu}\right\},
 \label{eq:shifted-rank-clock}
\end{equation}
where $\mu\in(0,1]$, to be optimized in the competitive
analysis.\footnote{The truncation is needed for
\cref{lem:monotone}: without it, $\lambda_t>1$ whenever $q_t>\mu$, and the
one-round charge is no longer nondecreasing.}  Thus the algorithm itself has no random coins; only the
input order is random.

\subsection{Random-order rank facts}

\begin{lemma}[First and later ranks]\label{lem:clock}
Fix a set $C$ of $k$ labeled demand-point copies in a uniformly random
permutation of $n$ copies, where $1\le k\le n$.  Let $T$ be the first rank
occupied by $C$, and let $p$ be the copy at that rank.  Then $p$ is uniform
in $C$ and independent of the set of ranks occupied by $C$.  For every
integer $j$ with $1\le j\le n-k+1$,
\begin{equation}
 \Prb(T\ge j)
 =
 \frac{\binom{n-j+1}{k}}{\binom nk}.
 \label{eq:first-rank-tail}
\end{equation}
Conditional on the ordered prefix through rank $T$ and on the identity of
$p$, the rank of each fixed $u\in C\setminus\{p\}$ is uniform on
$\{T+1,\ldots,n\}$.
\end{lemma}

\begin{proof}
The $k$ ranks occupied by $C$ form a uniformly random $k$-subset of $[n]$.
Equation~\eqref{eq:first-rank-tail} counts the subsets contained in
$\{j,\ldots,n\}$.  Conditional on any occupied rank set, all bijections from
the copies in $C$ to those ranks are equally likely; this proves that $p$ is uniform in $C$ and independent of the occupied rank set.  After the
ordered prefix through $T$ is exposed, the unrevealed copies are still in a
uniformly random order, which gives the conditional distribution of every
fixed later copy.
\end{proof}

\subsection{Penalized-Distance RankCut}

Fix once and for all a deterministic ordering of the finite candidate set
$\F$.  Every minimization below uses this ordering to break ties.  In
particular, the tie-breaking is independent of the current open set and of
the current connection distance.

When $x$ arrives at rank $t$, choose
\begin{equation}
 y_t(x)
 \in
 \arg\min_{y\in\F}
 \left\{d(x,y)+\lambda_t f_y\right\}.
 \label{eq:penalized-choice}
\end{equation}
Let $D=d(x,O)$ be the distance to the facilities already open.  The algorithm
opens $y_t(x)$ if
\begin{equation}
 D
 \ge
 d\bigl(x,y_t(x)\bigr)
 +\lambda_t f_{y_t(x)}.
 \label{eq:penalized-cutoff}
\end{equation}
It then connects $x$ to a nearest open facility.  Thus the choice in
\eqref{eq:penalized-choice} trades connection distance against a
rank-discounted opening cost, while \eqref{eq:penalized-cutoff} opens only
when the current connection distance covers that penalized objective.

\begin{algorithm}[H]
\caption{Penalized-Distance RankCut}\label{alg:main}
\begin{algorithmic}[1]
\State $O_0\gets\varnothing$.
\For{$t=1,\ldots,n$}
  \State Receive demand point $x=v_t$ and set
         $D\gets d(x,O_{t-1})$, $q_t\gets t/n$, and
         $\lambda_t\gets\min\{1,q_t/\mu\}$.
  \State Choose, with the fixed tie-breaking,
  \[
   y\in\arg\min_{z\in\F}
   \{d(x,z)+\lambda_t f_z\}.
  \]
  \State $O_t\gets O_{t-1}$.
  \If{$D\ge d(x,y)+\lambda_t f_y$}
    \State $O_t\gets O_t\cup\{y\}$.
  \EndIf
  \State Connect $x$ to its nearest facility in $O_t$.
\EndFor
\end{algorithmic}
\end{algorithm}

When $O_0=\varnothing$, the value $D=+\infty$ makes the first opening
automatic, so every demand point is served.  Because $f_y>0$ and
$q_t>0$, a candidate already in $O_{t-1}$ cannot satisfy
\eqref{eq:penalized-cutoff}: in that case
$D\le d(x,y)<d(x,y)+\lambda_t f_y$.  Hence every opening cost charged by
the algorithm corresponds to a newly opened facility.

\begin{proposition}[Exact unit-cost specialization]
\label{prop:unit-specialization}
Suppose that $f_y=1$ for every $y\in\F$, and put
\[
 a_x:=d(x,\F),
 \qquad
 \Delta(x;D):=(D-a_x)^+.
\]
Then Algorithm~\ref{alg:main} chooses a nearest candidate to $x$ and opens it
if and only if
\[
 \Delta(x;D)\ge\lambda_t.
\]
Consequently, in the unit-cost prescribed-candidate model,
Algorithm~\ref{alg:main} is exactly the shifted-rank Candidate Marginal
DistCut rule with $q_t=t/n$.
\end{proposition}

\begin{proof}
With unit costs, the additive term $\lambda_t f_y=\lambda_t$ is
independent of $y$.  The fixed tie-breaking in
\eqref{eq:penalized-choice} therefore selects a fixed nearest candidate, and
the minimum penalized objective is $a_x+\lambda_t$.  Thus
\eqref{eq:penalized-cutoff} is equivalent to
$D-a_x\ge\lambda_t$.  Since the right-hand side is nonnegative, this is
equivalent to the displayed marginal cutoff.
\end{proof}

The positive shift in \eqref{eq:shifted-rank-clock} is essential for arbitrary
opening costs.  The zero-start clock used by Huang and Jiang~\cite{HuangJiang2026} in the uniform
full-space model can be unbounded here.

\begin{proposition}[Failure of the zero-start clock for nonuniform costs]
\label{prop:zero-start-unbounded}
Replace $q_t=t/n$ in Algorithm~\ref{alg:main} by
$z_t=(t-1)/n$.  Even for one request and finite positive opening costs, the
competitive ratio of the resulting rule is unbounded.
\end{proposition}

\begin{proof}
Let $n=1$ and let the unique request be the candidate point $x=y_M$.  The
metric has the two points $y_M$ and $c$ at distance one, and their costs are
\[
 f_{y_M}=M,
 \qquad
 f_c=1.
\]
All opening costs are positive and finite.  At the first
rank, $z_1=0$, so the penalized objective ignores opening costs and uniquely
selects $y_M$.  Since the initial connection distance is infinite, the rule
opens $y_M$ and pays $M$.  For $M>2$, the optimum opens $c$ and pays $2$.
The ratio is at least $M/2$, which is unbounded as $M\to\infty$.
\end{proof}

In the unit-cost model this obstruction disappears because the opening-cost
penalty is common to all candidates.  Appendix~\ref{app:uniform-candidates}
therefore uses Huang and Jiang's zero-start normalization for the sharper
single-scale analysis.

\begin{remark}[One continuous clock, two deterministic discretizations]
\label{rem:continuous-clock}
Partition the continuous time interval $[0,1]$ into $n$ equal subintervals.  At rank $t$, the corresponding subinterval has endpoints
\[
  z_t=\frac{t-1}{n}
  \qquad\text{and}\qquad
  q_t=\frac{t}{n}.
\]
Thus $z_t$ and $q_t$ are, respectively, the left- and right-endpoint
discretizations of the same continuous clock.  Although they differ only in
the choice of endpoint, they define distinct deterministic clocks:
$z_1=0$, whereas $q_1=1/n>0$.  We use the right-endpoint convention in the
main nonuniform-cost algorithm and the left-endpoint convention in the
unit-cost analysis of Appendix~\ref{app:uniform-candidates}.  The guarantees
for the two conventions are established by their respective discrete
analyses.
\end{remark}

\begin{proposition}[Polynomial-time implementation]\label{prop:runtime}
Assume that $d(x,y)$ can be evaluated for each arriving $x$ and
$y\in\F$.  Algorithm~\ref{alg:main} uses $O(|\F|)$ distance evaluations and
arithmetic operations per demand point, and is therefore polynomial-time in
$n$ and $|\F|$.
\end{proposition}

\begin{proof}
One scan of the finite candidate set finds the minimizer in
\eqref{eq:penalized-choice}; a scan of the open subset computes the current
connection distance.  The fixed candidate ordering resolves all ties without
additional state.
\end{proof}

\section{Competitive analysis}\label{sec:analysis}

We analyze the algorithm on the clusters of a fixed optimal solution.  The
proof uses two deterministic tools.  A one-round upper charge is monotone in
the current connection distance, so the history seen by a later cluster point
can be replaced by a distance through the first cluster point. An upper-envelope decomposition then charges the extra cost of the
selected candidate one coordinate increment at a time.

\subsection{One-round monotonicity and a linear-penalty envelope}

Fix $x$ and rank $t$, and abbreviate
\[
 y=y_t(x),
 \qquad
 r=d(x,y),
 \qquad
 f=f_y,
 \qquad
 m=r+\lambda_t f.
\]
For a hypothetical current connection distance
$D\in\R_{\ge0}\cup\{+\infty\}$, define
\begin{equation}
 W_x(D,t)
 :=
 \begin{cases}
  D, & D<m,\\[1mm]
  r+f, & D\ge m.
 \end{cases}
 \label{eq:one-round-charge}
\end{equation}

\begin{lemma}[One-round monotonicity]\label{lem:monotone}
For fixed $x$ and $t$, the function $D\mapsto W_x(D,t)$ is
nondecreasing.  If $D$ is the algorithm's actual pre-service connection
distance, its opening plus connection cost on $x$ is at most
$W_x(D,t)$.
\end{lemma}

\begin{proof}
If $D<m$, the algorithm opens no facility and pays connection cost $D$.  If
$D\ge m$, it opens $y$ and pays at most $f+r$.  At the only crossing point,
\[
 (r+f)-m=(1-\lambda_t)f\ge0.
\]
The charge has slope one before the crossing and is constant afterward, so it
is nondecreasing and dominates the actual one-round cost.
\end{proof}

We next formalize the envelope fact used twice below.  The explicit treatment
of ties is useful because the algorithm itself uses deterministic
tie-breaking.

\begin{lemma}[Linear-penalty upper envelope]\label{lem:envelope}
Let $\mathcal P\subseteq\R_{\ge0}^2$ be finite, contain $(0,0)$, and have
positive first coordinate at every nonorigin point. Let
\[
 \mathcal Q:=\operatorname{conv}(\mathcal P).
\]
For each $z>0$, let 
\[
F_z
 :=
 \operatorname*{arg\,max}_{(a,b)\in\mathcal Q}\{b-za\},
 \]
and let $\Gamma:=\bigcup_{z>0}F_z$.
List, in increasing order of first coordinate, all points of
$\mathcal P\cap\Gamma$ as
\[
 v_0=(0,0),v_1,\ldots,v_m.
\]
Write
\[
 v_i=(a_i,b_i),
 \qquad
 w_i:=a_i-a_{i-1},
 \qquad
 g_i:=b_i-b_{i-1}.
\]
Then $w_i>0$, so the slopes
\[
 x_i:=\frac{g_i}{w_i}
\]
are well defined, and
\[
 g_i\ge0,
 \qquad
 x_1\ge x_2\ge\cdots\ge x_m.
\]

For any $z>0$, select a point of $\mathcal P$ maximizing $b-za$, using a
fixed tie-breaking rule.  If the selected point is $v_\ell$,
then
\[
 x_i\ge z\quad(i\le\ell),
 \qquad
 x_i\le z\quad(i>\ell).
\]
In particular, $v_\ell=\sum_{i=1}^{\ell}(w_i,g_i)$.
\end{lemma}
\begin{proof}
Maximizing $b-za$ over $\mathcal P$ is equivalent to maximizing it
over $\mathcal Q=\operatorname{conv}(\mathcal P)$, since every point
of $\mathcal Q$ is a convex combination of points in $\mathcal P$.

We first claim that $(0,0)\in\Gamma$.  Indeed, if
$\mathcal P\ne\{(0,0)\}$, choose $z>
 \max_{(a,b)\in\mathcal P\setminus\{(0,0)\}}\frac{b}{a}$. Then every nonorigin point has $b-za<0$, while the origin has value
zero.

The first coordinates of $v_0,\ldots,v_m$ are strictly increasing.
Otherwise, among two distinct points with the same first coordinate,
the one with smaller second coordinate could not maximize $b-za$ for
any $z>0$.  Hence $w_i>0$.

The second coordinates are nondecreasing.  Indeed, if
$a<a'$ and $b\ge b'$, then for every $z>0$,
\[
 b-za>b'-za',
\]
so $(a',b')$ cannot belong to $\Gamma$.  Therefore $g_i\ge0$.

We next show that the slopes are nonincreasing.  Suppose
$x_i<x_{i+1}$ for some $i$.  Let
\[
 \alpha:=\frac{w_{i+1}}{w_i+w_{i+1}}.
\]
Then
\[
 a_i=\alpha a_{i-1}+(1-\alpha)a_{i+1},
\]
and
\[
 \alpha b_{i-1}+(1-\alpha)b_{i+1}-b_i
 =
 \frac{w_iw_{i+1}}{w_i+w_{i+1}}
 (x_{i+1}-x_i)
 >0.
\]
Thus, for every $z>0$,
\[
 b_i-za_i
 <
 \alpha(b_{i-1}-za_{i-1})
 +(1-\alpha)(b_{i+1}-za_{i+1}),
\]
so $v_i$ cannot maximize $b-za$, contradicting $v_i\in\Gamma$.
Hence
\[
 x_1\ge x_2\ge\cdots\ge x_m.
\]

Finally, fix $z>0$ and suppose that the selected maximizer is
$v_\ell$.  Moving from $v_{i-1}$ to $v_i$ changes the objective by
\[
 (b_i-za_i)-(b_{i-1}-za_{i-1})
 =
 w_i(x_i-z).
\]
If $x_i<z$ for some $i\le\ell$, then all slopes
$x_i,\ldots,x_\ell$ are smaller than $z$, so the objective strictly
decreases from $v_{i-1}$ to $v_\ell$, contradicting optimality.
Thus $x_i\ge z$ for every $i\le\ell$.

Similarly, if $x_i>z$ for some $i>\ell$, then all slopes
$x_{\ell+1},\ldots,x_i$ are larger than $z$, so the objective strictly
increases from $v_\ell$ to $v_i$, again a contradiction.  Hence
$x_i\le z$ for every $i>\ell$.  Both conclusions are non-strict and therefore hold for every maximizing
retained point, independently of how ties are resolved. Finally,
\[
 v_\ell
 =
 v_0+\sum_{i=1}^{\ell}(v_i-v_{i-1})
 =
 \sum_{i=1}^{\ell}(w_i,g_i).
\]
\end{proof}

In the two applications below, the map from candidate facilities to
points of $\mathcal P$ need not be injective.  We interpret
$\mathcal P$ as the set of distinct coordinate pairs: if several
candidates induce the same pair $(a,b)$, they correspond to a single
point of $\mathcal P$.  The algorithm's fixed candidate ordering still
selects one of the underlying candidate identities, but the selected
coordinate pair is the same.  Since the envelope decomposition and all
subsequent bounds depend only on the coordinates $(a,b)$, this choice
does not affect the argument.

We also add the auxiliary point $(0,0)$ to
the points associated with the candidates under consideration.  This
auxiliary point does not represent a facility and always has objective
value zero.  Ties among candidate points are resolved by the
algorithm's fixed ordering, and a candidate point is preferred to
$(0,0)$ when their objective values are equal.  If all candidate
points under consideration have negative objective value, $(0,0)$ is
selected.  This convention is used in the proofs of
\cref{lem:suffix-one,lem:high}: in the former, selecting $(0,0)$
represents zero excess over the auxiliary value $B$, while in the
latter it represents $H_p=0$.

\subsection{A cluster of the optimal solution}\label{sec:cluster-analysis}

Fix an optimal offline solution. Remove every opened facility to which no demand copy is assigned.  Since all
opening costs are positive, this does not increase the optimum.  Hence every
cluster considered below is nonempty.  Consider a facility $c\in\F$ opened by
this solution, and let $C$ be the multiset of demand-point copies assigned to
$c$.  Write
\[
 k:=|C| \ge 1,
 \qquad
 f:=f_c,
 \qquad
 r_u:=d(u,c),
 \qquad
 R(C):=\sum_{u\in C}r_u.
\]
The optimal cost associated with this cluster is $f+R(C)$.

Let $p$ be the first copy in $C$ and let $T$ be its rank.  Let $D_p$ be the
connection distance paid after $p$ is served.  By \cref{lem:clock}, $p$ is
uniform over the $k$ copies, and therefore
\[
 \E[r_p]=\frac{R(C)}k.
\]

For a demand copy $u$, let $\ALG(u)$ denote its connection cost together
with the opening cost, if any, of the facility whose opening is triggered
when $u$ is served.  For an optimal cluster $C$, define
\[
 \ALG(C):=\sum_{u\in C}\ALG(u).
\]
Thus every online opening cost is assigned to the optimal cluster containing
the demand that triggers the opening.  Every opening cost and every
connection cost is counted exactly once, and hence
\[
 \ALG=\sum_C\ALG(C).
\]

\subsection{Later demand points}

Let $\mathcal G_p$ denote all information revealed by the time $p$ has been
served: the identity and rank $T$ of $p$, the ordered request prefix through
rank $T$, and the algorithm's open set and decisions up to that time, under
the fixed tie-breaking rule.  Formally, $\mathcal G_p$ is the sigma-field
generated by these random objects.  It exposes
no position of a later copy in $C$.  For every fixed
$u\in C\setminus\{p\}$, conditional on $\mathcal G_p$, its rank $J$ is uniform
on $\{T+1,\ldots,n\}$.  Indeed, before $p$ no other point of $C$ has arrived,
and the algorithm's deterministic decisions on the exposed prefix reveal no
additional information about the remaining permutation.

\begin{lemma}[Cost after the first cluster arrival]\label{lem:suffix-one}
Conditioned on $\mathcal G_p$, for every fixed
$u\in C\setminus\{p\}$,
\[
 \E[\ALG(u)\mid\mathcal G_p]
 \le
 (1+\mu)\bigl(D_p+d(p,u)\bigr).
\]
\end{lemma}

\begin{proof}
Set
\[
 B:=D_p+d(p,u).
\]
Choose a facility $y_p$ that serves $p$ after its opening decision, so
$d(p,y_p)=D_p$.  Facilities never close, and therefore the actual distance
$D_u^-$ immediately before $u$ is served satisfies
\[
 D_u^-
 \le d(u,y_p)
 \le d(u,p)+d(p,y_p)
 =B.
\]
By \cref{lem:monotone},
\begin{equation}
 \ALG(u)
 \le W_u(D_u^-,J)
 \le W_u(B,J).
 \label{eq:later-monotone}
\end{equation}
The value $W_u(B,J)$ is only an auxiliary charge; no open set with
connection distance exactly $B$ is required.

For each $y\in\F$, define the improvement relative to $B$ by
\[
 P_y:=(B-d(u,y))^+.
\]
If $W_u(B,J)\le B$, then
\[
 \bigl(W_u(B,J)-B\bigr)^+=0.
\]
Otherwise, let $y$ be the candidate selected by
\eqref{eq:penalized-choice}.  Since $W_u(B,J)>B$, the opening branch in
\eqref{eq:one-round-charge} applies, and hence
\[
 B\ge d(u,y)+\lambda_J f_y.
\]
It follows that
\[
 P_y-\lambda_J f_y
 =
 B-d(u,y)-\lambda_J f_y
 \ge0.
\]

Moreover, $y$ maximizes $P_{y'}-\lambda_J f_{y'}$ over all
$y'\in\F$, together with the auxiliary point $(0,0)$.  Indeed, if
$d(u,y')\le B$, then
\[
 P_{y'}-\lambda_J f_{y'}
 =
 B-\bigl(d(u,y')+\lambda_J f_{y'}\bigr),
\]
whereas if $d(u,y')>B$, then
\[
 P_{y'}-\lambda_J f_{y'}
 =
 -\lambda_J f_{y'}<0.
\]
The selected candidate has a nonnegative objective value, so neither
the latter candidates nor the auxiliary point can strictly improve
upon it.

Apply \cref{lem:envelope} to
\[
 \{(f_y,P_y):y\in\F\}\cup\{(0,0)\}.
\]
Let $v_0,\ldots,v_m$ and the increments $(w_i,g_i)$ be as in the
lemma.  If the selected point is
\[
 v_\ell=(f_y,P_y),
\]
then
\[
 (f_y,P_y)
 =
 \sum_{i=1}^{\ell}(w_i,g_i),
\]
and
\[
 x_i\ge\lambda_J
 \qquad(i\le\ell).
\]
Consequently,
\[
 f_y-P_y
 =
 \sum_{i=1}^{\ell}(w_i-g_i)
 =
 \sum_{i=1}^{\ell}w_i(1-x_i).
\]
Combining this identity with the case $W_u(B,J) \le B$ gives
\begin{align}
 \bigl(W_u(B,J)-B\bigr)^+
 &\le
 \sum_i
 w_i(1-x_i)^+
 \1\{\lambda_J\le x_i\}.
 \label{eq:later-increment-bound}
\end{align}

Only indices with $x_i<1$ can give a nonzero term.  For such an
index,
\[
 \lambda_J\le x_i
 \quad\Longleftrightarrow\quad
 \frac{J}{n}\le\mu x_i.
\]
Conditional on $\mathcal G_p$, at most
\[
 (n\mu x_i-T)^+
\]
ranks in $\{T+1,\ldots,n\}$ satisfy this inequality.  Hence the
expected value of the $i$th term in
\eqref{eq:later-increment-bound} is at most
\[
 w_i(1-x_i)
 \frac{(n\mu x_i-T)^+}{n-T}
 =
 w_i(1-x_i)
 \frac{(\mu x_i-q_T)^+}{1-q_T}.
\]
If this expression is nonzero, then
$q_T<\mu x_i\le x_i$, where the second inequality uses $\mu\le1$.
Thus
\[
 \frac{1-x_i}{1-q_T}\le1,
\]
and the expected value of this term is at most
\[
 w_i(\mu x_i-q_T)^+
 \le
 \mu w_ix_i
 =
 \mu g_i.
\]

Finally, since
\[
 v_m=\sum_i(w_i,g_i),
\]
the quantity $\sum_i g_i$ is the second coordinate of $v_m$.
Every point in the set to which \cref{lem:envelope} is applied has
second coordinate at most $\max_y P_y$, and therefore
\[
 \sum_i g_i
 \le
 \max_y P_y
 \le
 B.
\]
It follows that
\[
 \E[W_u(B,J)\mid\mathcal G_p]
 \le
 B+\mu\sum_i g_i
 \le
 (1+\mu)B.
\]
Together with \eqref{eq:later-monotone}, this proves the lemma.
\end{proof}

No independence among the costs of different later points is used.  The
cluster proof below invokes \cref{lem:suffix-one} only through conditional
linearity and the tower property.

\subsection{The first demand point}

Define
\[
 a_k:=1+(k-1)(1+\mu)=k+(k-1)\mu.
\]
Let $y=y_T(p)$ be the candidate chosen by
\eqref{eq:penalized-choice} when $p$ arrives.  Since the optimal facility $c$
is a candidate,
\begin{equation}
 d(p,y)+\lambda_T f_y
 \le r_p+\lambda_T f.
 \label{eq:first-choice}
\end{equation}
Define the anchor excess
\[
 e:=(D_p-r_p)^+.
\]
If $p$ opens $y$, split its opening cost as
\[
 L_p:=\min\{f_y,f\},
 \qquad
 H_p:=(f_y-f)^+.
\]
If $p$ opens no facility, set $L_p=H_p=0$.  Thus the opening cost triggered
by $p$ is exactly $L_p+H_p$.

\begin{lemma}[Pointwise base-cost certificate]\label{lem:first-base}
For every realization,
\[
 L_p+a_ke
 \le
 f\max\{1,a_k\lambda_T\}.
\]
\end{lemma}

\begin{proof}
There are three cases.

If $p$ opens no facility, the failed cutoff and
\eqref{eq:first-choice} imply
\[
 D_p<d(p,y)+\lambda_T f_y
 \le r_p+\lambda_T f.
\]
Hence $e\le\lambda_T f$, and the claim follows because $L_p=0$.

Suppose that $p$ opens $y$ and $f_y\le f$.  Then
$D_p\le d(p,y)$ and \eqref{eq:first-choice} gives
\[
 e\le\lambda_T(f-f_y).
\]
Consequently,
\[
 L_p+a_ke
 \le f_y+a_k\lambda_T(f-f_y).
\]
The right-hand side is affine in $f_y\in[0,f]$, so its maximum is attained at
an endpoint and equals at most
$f\max\{1,a_k\lambda_T\}$.

Finally, suppose that $p$ opens $y$ and $f_y>f$.  Since $q_T>0$,
$\lambda_T>0$, and \eqref{eq:first-choice} yields
$d(p,y)<r_p$.  Thus $e=0$ and $L_p=f$.
\end{proof}

\begin{lemma}[Expected base-cost certificate]\label{lem:first-expectation}
Define
\begin{equation}
 \Psi(\mu)
 :=
 1+\frac{1+\mu}{\mu}
 \exp\!\left(-\frac{\mu}{1+\mu}\right).
 \label{eq:Psi}
\end{equation}
Then
\[
 \E[L_p+a_ke]\le\Psi(\mu)f.
\]
\end{lemma}

\begin{proof}
Put
\[
 \rho:=\frac{\mu}{1+\mu},
 \qquad
 \gamma_k:=\frac{a_k}{n\mu}=\frac{k-\rho}{n\rho},
 \qquad
 \ell:=\left\lfloor\frac{1}{\gamma_{k}}\right\rfloor.
\]
Since $a_k=(1+\mu)(k-\rho)$,
\[
 h(T):=\max\{1,a_k\lambda_T\}
 =\max\{1,\min\{a_k,\gamma_k T\}\}.
\]
For every integer $T\ge1$,
\begin{equation}
 h(T)
 \le
 1+\gamma_k (T-\ell)^+-(1-\gamma_k \ell)\1\{T>\ell\}.
 \label{eq:discrete-first-pointwise}
\end{equation}
Indeed, if $T\le\ell$ then $\gamma_k T\le1$ and both sides equal one; if
$T>\ell$, the right-hand side is $\gamma_k T>1$ and dominates
$\min\{a_k,\gamma_k T\}$.

The first cluster rank satisfies $T\le n-k+1$.  If
$\ell\ge n-k+1$, then $h(T)=1$ and the result is immediate.  Otherwise, let
\[
 P_\ell:=\Prb(T>\ell)
 =\frac{\binom{n-\ell}{k}}{\binom nk}.
\]
The tail-sum formula and the hockey-stick identity give
\begin{align}
 \E[(T-\ell)^+]
 &=\sum_{j=\ell+1}^{n-k+1}\Prb(T\ge j)\notag\\
 &=\frac{\binom{n-\ell+1}{k+1}}{\binom nk}
 =P_\ell\frac{n-\ell+1}{k+1}.
 \label{eq:first-rank-overshoot}
\end{align}
Taking expectations in \eqref{eq:discrete-first-pointwise} yields
\[
 \E[h(T)]\le1+P_\ell B,
 \qquad
 B:=\gamma_k \frac{n-\ell+1}{k+1}-1+\gamma_k \ell.
\]
Because we are in the case $\ell<n-k+1$, integrality gives
$\ell\le n-k$, and hence
$\frac{n-\ell+1}{k+1}\ge1$.
Moreover, the definition
$\ell=\lfloor1/\gamma_k\rfloor$ implies
$\gamma_k(\ell+1)>1$.
Therefore,
\[
 B
 \ge
 \gamma_k-1+\gamma_k\ell
 =
 \gamma_k(\ell+1)-1
 >0.
\]

Define
\[
 s:=\frac{(k-\rho)\ell}{n}.
\]
The definition of $\ell$ gives $0\le s\le\rho$.  Moreover,
\[
 B=
 \frac{(k-\rho)(n-\ell+1)}{n\rho(k+1)}
 -1+\frac{s}{\rho},
\]
and the first fraction is at most $1/\rho$.  Hence
\begin{equation}
 B\le\frac{1-\rho+s}{\rho}.
 \label{eq:first-rank-B}
\end{equation}
Moreover,
\[
 P_\ell
 =\prod_{j=0}^{\ell-1}\left(1-\frac{k}{n-j}\right)
 \le\exp\!\left(-\frac{k\ell}{n}\right)
 \le e^{-s}.
\]
Combining the positivity of $B$ with
\eqref{eq:first-rank-B} gives
\[
 P_\ell B
 \le\frac{(1-\rho+s)e^{-s}}{\rho}
 \le\frac{e^{-\rho}}{\rho},
\]
because the derivative of $(1-\rho+s)e^{-s}$ is
$(\rho-s)e^{-s}\ge0$ on $[0,\rho]$.  Consequently,
\[
 \E[h(T)]
 \le1+\frac{e^{-\rho}}{\rho}
 =1+\frac{1+\mu}{\mu}
   \exp\!\left(-\frac{\mu}{1+\mu}\right)
 =\Psi(\mu).
\]
Combining this with \cref{lem:first-base} proves the lemma.
\end{proof}

\subsection{High-cost excess at the first point}

The remaining part $H_p$ is the amount by which the facility opened at $p$
costs more than the optimal cluster facility.

\begin{lemma}[High-cost excess]\label{lem:high}
For a cluster $C$ of size $k$,
\[
 \E[H_p]
 \le
 \max\left\{\frac1k,\mu\right\}R(C)
 \le
 \left(\frac1k+\mu\right)R(C).
\]
\end{lemma}

\begin{proof}
Fix the identity of $p$.  By \cref{lem:clock}, the identity of $p$ is
independent of the set of ranks occupied by $C$.  Hence, under this
conditioning, the occupied rank set remains a uniformly random
$k$-subset of $[n]$.  All probabilities below are taken under this
conditioning.

For every candidate with $f_y>f$, define
\[
 s_y:=f_y-f,
 \qquad
 b_y:=(r_p-d(p,y))^+.
\]
Consider the event $H_p>0$, on which the algorithm opens a candidate
$y$ with $f_y>f$.  On this event,
\eqref{eq:first-choice} implies
\[
 b_y-\lambda_Ts_y\ge0.
\]

The selected candidate $y$ maximizes
\[
 b_{y'}-\lambda_Ts_{y'}
\]
over all candidates $y'$ with $f_{y'}>f$, together with the auxiliary
point $(0,0)$.  Indeed, if $d(p,y')\le r_p$, then
\[
 b_{y'}-\lambda_Ts_{y'}
 =
 r_p+\lambda_Tf
 -\bigl(d(p,y')+\lambda_Tf_{y'}\bigr),
\]
whereas if $d(p,y')>r_p$, then
\[
 b_{y'}-\lambda_Ts_{y'}
 =
 -\lambda_Ts_{y'}<0.
\]
Thus a candidate of the latter type cannot improve upon the selected
candidate, whose objective value is nonnegative.

Apply \cref{lem:envelope}, with the auxiliary-point tie-breaking
specified above, to
\[
 \{(s_y,b_y):f_y>f\}\cup\{(0,0)\}.
\]
Let $v_0,\ldots,v_m$ and the increments $(w_i,g_i)$ be as in the
lemma.  On the event $H_p>0$, the selected point is
\[
 v_\ell=(s_y,b_y).
\]
Then
\[
 H_p=s_y=\sum_{i=1}^{\ell}w_i
\]
and
\[
 x_i\ge\lambda_T
 \qquad(i\le\ell).
\]
Therefore,
\begin{equation}
 H_p
 \le
 \sum_i w_i\1\{\lambda_T\le x_i\}.
 \label{eq:high-increment-bound}
\end{equation}
When no high-cost candidate is opened, $H_p=0$, so the same inequality
remains valid.

If $x_i<1$, then
\[
 \Prb(\lambda_T\le x_i)
 =
 \Prb\!\left(T\le\lfloor n\mu x_i\rfloor\right)
 \le
 k\frac{\lfloor n\mu x_i\rfloor}{n}
 \le
 k\mu x_i,
\]
where the first inequality is a union bound over the $k$ cluster
copies.  Hence the expected contribution of the $i$th term in
\eqref{eq:high-increment-bound} is at most
\[
 k\mu x_iw_i
 =
 k\mu g_i.
\]
If $x_i\ge1$, then
\[
 w_i\le g_i,
\]
so the contribution of the $i$th term is at most $g_i$.  Therefore,
conditional on the identity of $p$,
\[
 \E[H_p\mid p]
 \le
 \max\{1,k\mu\}\sum_i g_i.
\]

Since
\[
 v_m=\sum_i(w_i,g_i),
\]
the quantity $\sum_i g_i$ is the second coordinate of $v_m$.
Every point in the set to which \cref{lem:envelope} is applied has
second coordinate at most
\[
 \max\bigl(\{b_y:y\in\F,\ f_y>f\}\cup\{0\}\bigr).
\]
Hence
\[
 \sum_i g_i
 \le
 \max\bigl(\{b_y:y\in\F,\ f_y>f\}\cup\{0\}\bigr)
 \le
 r_p.
\]
Consequently
\[
 \E[H_p\mid p]
 \le
 \max\{1,k\mu\}r_p.
\]
The identity of $p$ is uniform in $C$, so
$\E[r_p]=R(C)/k$.  Averaging over $p$ proves the first inequality; the second
uses $\max\{1/k,\mu\}\le1/k+\mu$.
\end{proof}

\subsection{The cluster and global bounds}

\begin{lemma}[Cluster bound]\label{lem:cluster}
For every optimal cluster $C$ served by a facility $c$ of cost $f_c$,
\[
 \E[\ALG(C)]
 \le
 \Psi(\mu)f_c+(3+4\mu)R(C).
\]
\end{lemma}

\begin{proof}
Let
\[
 Y_C:=\sum_{u\in C}\1\{u\ne p\}\ALG(u)
\]
be the cost incurred on the points of $C$ arriving after $p$.  The identity
of $p$ is $\mathcal G_p$-measurable.  Conditional linearity and
\cref{lem:suffix-one} yield
\begin{align*}
 \E[Y_C\mid\mathcal G_p]
 &\le
 (k-1)(1+\mu)D_p
 +(1+\mu)\sum_{u\ne p}d(p,u).
\end{align*}
The quantities $L_p,H_p$, and $D_p$ are determined when $p$ has been served.
Adding the opening and connection cost of $p$ and applying the tower property
gives
\begin{align}
 \E[\ALG(C)]
 \le
 \E\!\left[
  L_p+H_p+a_kD_p
  +(1+\mu)\sum_{u\ne p}d(p,u)
 \right].
 \label{eq:cluster-tower}
\end{align}
No independence among later costs is used.

Since $D_p\le r_p+e$, \cref{lem:first-expectation} controls
$L_p+a_ke$.  Also,
\[
 \E[a_kr_p]
 =\frac{a_k}{k}R(C)
 =\left(1+\mu-\frac{\mu}{k}\right)R(C).
\]
By the triangle inequality,
\[
 \sum_{u\ne p}d(p,u)
 \le
 \sum_{u\ne p}(r_p+r_u)
 =R(C)+(k-2)r_p,
\]
and therefore
\[
 \E\!\left[\sum_{u\ne p}d(p,u)\right]
 \le2\left(1-\frac1k\right)R(C).
\]
Using the weaker form of \cref{lem:high}, the coefficient of $R(C)$ in
\eqref{eq:cluster-tower} is at most
\begin{align*}
 &\left(\mu+\frac1k\right)
 +\left(1+\mu-\frac{\mu}{k}\right)
 +2(1+\mu)\left(1-\frac1k\right)\\
 &=3+4\mu-\frac{1+3\mu}{k}
 \le3+4\mu.
\end{align*}
This proves the lemma.
\end{proof}

\subsection{Completing the upper bound}\label{sec:global}

\begin{theorem}[General parameter bound]\label{thm:general}
For every $\mu\in(0,1]$, Algorithm~\ref{alg:main} is a deterministic
polynomial-time online algorithm for known horizon $n$ and satisfies
\[
 \E[\ALG]
 \le
 \max\{\Psi(\mu),3+4\mu\}\OPT,
\]
where $\Psi$ is defined in \eqref{eq:Psi}.
\end{theorem}

\begin{proof}
Let $F^*$ and $R^*$ be the opening and connection costs of a fixed optimal
solution.  Sum \cref{lem:cluster} over its clusters.  The bookkeeping defined
in \cref{sec:cluster-analysis} counts each online opening and connection cost
exactly once, so
\[
 \E[\ALG]
 \le
 \Psi(\mu)F^*+(3+4\mu)R^*.
\]
Since $\OPT=F^*+R^*$, the claimed maximum follows.
\end{proof}

A direct differentiation shows that $\Psi(\mu)$ is strictly decreasing,
whereas $3+4\mu$ is strictly increasing.  Hence the bound in
\cref{thm:general} is minimized at their unique intersection, which occurs
near
\[
 \mu=0.31684
\]
and gives a value approximately equal to $4.26736$.

\begin{corollary}[A $4.2674$ bound]\label{cor:constant}
With $\mu=0.31684$, Algorithm~\ref{alg:main} satisfies
\[
 \E[\ALG]<4.2674\,\OPT.
\]
\end{corollary}

\begin{proof}
Substitution gives
\[
 \Psi(0.31684)<4.26738,
 \qquad
 3+4(0.31684)=4.26736.
\]
The claim follows from \cref{thm:general}.
\end{proof}

\section{\texorpdfstring{A $3-o(1)$ lower bound and a model separation}{A 3-o(1) lower bound and a model separation}}
\label{sec:lower}

We first prove a factor-three lower bound in the restricted candidate-site
model, even though all candidate facilities have the same opening cost.  We
then transfer the construction to the classical full-space model with
nonuniform opening costs.  As in several random-order lower bounds, it is convenient
to begin in the weaker i.i.d. model, in which the algorithm is even given the
demand distribution in advance.

\begin{theorem}[A $3-o(1)$ random-order lower bound]
\label{thm:lower-three}
For every integer $N\ge2$ and every randomized online algorithm, there exists
a demand-point multiset of size $N$ in the metric of \cref{lem:lb-metric} such that
\[
 \frac{\E[\ALG]}{\OPT}
 \ge
 \frac{3N}{N+3},
\]
where the expectation is over the uniformly random arrival order and the
algorithm's internal randomness. The instance has $m=2N(N-1)$
possible demand locations and $\binom{m}{N}$
uniform-cost candidate facilities. In particular, the asymptotic competitive
ratio is at least $3$, even when all candidate facilities have the same
opening cost.
\end{theorem}

The construction hides a facility tailored to the realized support.  Offline
can choose this facility after observing all demand points.  Online, however,
must choose facilities adaptively, and each opened facility covers only a
limited expected number of future demand points.

\subsection{The distributional construction}

Fix an integer horizon $N\ge2$, and let
\[
 m:=2N(N-1).
\]
Let
\[
 \mathcal X:=\{x_1,\ldots,x_m\}
\]
be the set of possible demand locations. For every set
$I\subseteq[m]$ with $|I|=N$, introduce a distinct candidate facility
$y_I$. Let
\[
 \mathcal Y
 :=
 \{y_I:I\subseteq[m],\ |I|=N\}.
\]
The metric space of the lower-bound instance is
$\mathcal X\cup\mathcal Y$, the candidate-facility set is $\mathcal Y$, and
every candidate facility has opening cost one.

Define a distance function on $\mathcal X\cup\mathcal Y$ by
\begin{align*}
 d(x_i,x_j)&=\frac23 &&(i\ne j),\\
 d(y_I,y_{I'})&=\frac23 &&(I\ne I'),\\
 d(x_i,y_I)&=
 \begin{cases}
  \frac13,&i\in I,\\
  1,&i\notin I,
 \end{cases}
\end{align*}
with $d(z,z)=0$ and symmetry.

\begin{lemma}[Metric validity]\label{lem:lb-metric}
The function $d$ defined above is a metric on
$\mathcal X\cup\mathcal Y$.
\end{lemma}

\begin{proof}
Positivity and symmetry are immediate. It remains to verify the triangle
inequality. A triangle consisting only of demand locations, or only of facility
locations, is equilateral with side length $2/3$.

Consider two demand locations $x_i,x_j$ and one facility $y_I$. If $i,j\in I$, the
three side lengths are $1/3,1/3,2/3$. If exactly one of $i,j$ lies in $I$,
the lengths are $1/3,2/3,1$. If neither lies in $I$, the lengths are
$2/3,1,1$. Each case satisfies the triangle inequality.

For a triangle consisting of one demand location $x_i$ and two facilities
$y_I,y_{I'}$, the possible side lengths are again
\[
 \left(\frac13,\frac13,\frac23\right),
 \qquad
 \left(\frac13,1,\frac23\right),
 \qquad
 \left(1,1,\frac23\right),
\]
according to whether $i$ belongs to both, exactly one, or neither of
$I$ and $I'$. Thus the triangle inequality holds in every case.
\end{proof}

\paragraph{Demand distribution and offline optimum.}

Let the $N$ demand points be drawn independently and uniformly from
$\mathcal X$. The online algorithm is given the entire metric, the candidate
set, the opening costs, the horizon, and the demand distribution.

\begin{lemma}[Exact offline cost]\label{lem:lb-opt}
For every realization of the $N$ demand points,
\[
 \OPT=1+\frac N3.
\]
\end{lemma}

\begin{proof}
Let $S\subseteq[m]$ be the set of distinct indices appearing in the realized
demand-point sequence. Since $|S|\le N$ and $m\ge N$, extend $S$ to a set
$I\subseteq[m]$ of size exactly $N$. Opening $y_I$ costs one and connects
every demand point at distance $1/3$. Hence
\[
 \OPT\le1+\frac N3.
\]

Conversely, every feasible solution must open at least one candidate facility,
and the distance from every demand location to every candidate facility is at
least $1/3$. Hence every feasible solution has cost at least
\[
 1+\frac N3.
\]
The two bounds prove the claim.
\end{proof}

\subsection{The online cost}

Fix an arbitrary randomized online algorithm. All expectations below include
both the i.i.d.\ demand points and the internal randomness of the algorithm.

A demand point at location $x_i$ is called \emph{covered} at its service time if,
after the algorithm has made its opening decisions in that round, some open
facility $y_I$ satisfies $i\in I$. Let
\[
 K:=\text{the number of distinct facilities opened},
 \qquad
 C:=\text{the number of covered rounds}.
\]
A covered demand point has connection cost $1/3$, whereas an uncovered demand point has
connection cost $1$. Since every facility has opening cost one, we have the
pointwise identity
\begin{equation}\label{eq:lb-cost-identity}
 \ALG=N+K-\frac23C.
\end{equation}

\begin{lemma}[Expected coverage per opening]\label{lem:lb-coverage}
For every randomized online algorithm,
\[
 \E[C]\le\frac32\E[K].
\]
\end{lemma}

\begin{proof}
Assign each covered demand point one coverage credit, and assign that credit to one
facility that is open at the time the demand point is served and contains the
demand point's index. If several such facilities exist, use any fixed tie-breaking
rule. Thus every covered demand point contributes exactly one credit.

For every round $t$ and every facility $y_I$, let $A_{t,I}$ be the event that
$y_I$ is first opened in round $t$. Let $B_I$ denote the total number of
credits assigned to $y_I$, and define
\[
 Z_{t,I}:=\1_{A_{t,I}}B_I.
\]
Each opened facility has a unique first-opening round, and unopened facilities
receive no credit. Therefore, pointwise,
\[
 C=\sum_{t=1}^{N}\sum_I Z_{t,I}.
\]

Suppose that $A_{t,I}$ occurs. The facility $y_I$ can receive at most one
credit from the demand point served in round $t$. Any other credit assigned to
$y_I$ must come from a future demand point whose index lies in $I$. Hence
\[
 Z_{t,I}
 \le
 \1_{A_{t,I}}
 \left(
  1+\sum_{s=t+1}^{N}
  \1\{v_s\in\{x_i:i\in I\}\}
 \right).
\]

The event $A_{t,I}$ is measurable with respect to the history through the end
of round $t$. Conditional on that history, the future demand points remain
independent and uniform on the $m$ possible demand locations. Since $|I|=N$,
each future demand point lies in $I$ with probability $N/m$. Therefore,
\begin{align*}
 \E[Z_{t,I}]
 &\le
 \Prb(A_{t,I})
 \left(
  1+(N-t)\frac{N}{m}
 \right)\\
 &\le
 \Prb(A_{t,I})
 \left(
  1+\frac{N(N-1)}{m}
 \right)\\
 &=
 \frac32\Prb(A_{t,I}),
\end{align*}
where the last equality uses
\[
 m=2N(N-1).
\]

For each fixed facility $y_I$, the events
$A_{1,I},\ldots,A_{N,I}$ are mutually exclusive, and their union is the event
that $y_I$ is ever opened. Therefore, pointwise,
\[
 K=\sum_{t=1}^{N}\sum_I\1_{A_{t,I}},
\]
and hence
\[
 \sum_{t=1}^{N}\sum_I\Prb(A_{t,I})=\E[K].
\]
Summing the credit bounds over all rounds and facilities now gives
\begin{align*}
 \E[C]
 &=
 \sum_{t=1}^{N}\sum_I\E[Z_{t,I}]\\
 &\le
 \frac32\sum_{t=1}^{N}\sum_I\Prb(A_{t,I})\\
 &=
 \frac32\E[K].
\end{align*}
The argument permits the algorithm to open several facilities in the same
round: the current demand point is credited to only one of them, while every
distinct opened facility is counted exactly once through its unique
first-opening event.
\end{proof}

\begin{lemma}[Distributional online cost]\label{lem:lb-online}
Every randomized online algorithm satisfies
\[
 \E[\ALG]\ge N
\]
under the i.i.d.\ demand distribution above.
\end{lemma}

\begin{proof}
Taking expectations in \eqref{eq:lb-cost-identity} and applying
\cref{lem:lb-coverage}, we obtain
\begin{align*}
 \E[\ALG]
 &=
 N+\E[K]-\frac23\E[C]\\
 &\ge
 N+\E[K]-\frac23\cdot\frac32\E[K]\\
 &=
 N.
\end{align*}
\end{proof}

\subsection{From i.i.d. arrivals to random order}

\begin{lemma}[Conditioning on the demand-point multiset]\label{lem:iid-to-ro}
Let an ordered sequence of demand points be drawn i.i.d.\ from any distribution,
and let $U$ be the resulting demand-point multiset. Conditional on $U$, the ordered
sequence is a uniformly random ordering of the demand-point copies in $U$.

Consequently, if an online algorithm is $c$-competitive for every fixed
demand-point multiset under uniformly random arrival order, then it is also
$c$-competitive in the corresponding i.i.d.\ demand experiment.
\end{lemma}

\begin{proof}
Conditional on the multiplicity of every demand location, all ordered sequences
with those multiplicities have the same product probability. Thus the
conditional order is uniform.

If an algorithm were $c$-competitive for every fixed demand-point multiset under
uniformly random order, then for every realized multiset $U$,
\[
 \E[\ALG\mid U]\le c\,\OPT(U).
\]
Taking expectation over $U$ would give
\[
 \E[\ALG]\le c\,\E[\OPT(U)]
\]
in the i.i.d.\ experiment.
\end{proof}

\begin{proof}[Proof of \cref{thm:lower-three}]
Under the i.i.d.\ demand distribution, \cref{lem:lb-online} gives
\[
 \E[\ALG]\ge N,
\]
whereas \cref{lem:lb-opt} gives
\[
 \OPT(U)=1+\frac N3
\]
for every realized demand-point multiset $U$. Hence
\[
 \frac{\E[\ALG]}{\E[\OPT]}
 \ge
 \frac{N}{1+N/3}
 =
 \frac{3N}{N+3}.
\]

Suppose, for contradiction, that every realized demand-point multiset $U$ satisfied
\[
 \frac{\E[\ALG\mid U]}{\OPT(U)}
 <
 \frac{3N}{N+3},
\]
where the conditional expectation is over the uniformly random ordering of
the copies in $U$ and the algorithm's internal randomness. By
\cref{lem:iid-to-ro}, averaging this strict inequality over $U$ would give
\[
 \E[\ALG]<N,
\]
contradicting \cref{lem:lb-online}. Therefore at least one demand-point multiset has
the claimed random-order ratio. Letting $N\to\infty$ proves the asymptotic
lower bound.
\end{proof}

\subsection{The full-space nonuniform model}

The preceding construction uses a restricted candidate set: the demand
locations $x_i$ cannot themselves be opened.  We now show that this
restriction is not responsible for the factor-three lower bound.  The same
metric yields the same lower bound when every point is a feasible facility,
by assigning nonuniform opening costs.

\begin{corollary}[Full-space nonuniform lower bound]
\label{cor:fullspace-lower}
For every integer $N\ge2$ and every randomized online algorithm, there is a
finite metric space in which every metric point is a feasible facility
location, together with nonuniform opening costs and a demand-point multiset
of size $N$, such that
\[
 \frac{\E[\ALG]}{\OPT}
 \ge
 \frac{3N}{N+3}.
\]
Consequently, the asymptotic random-order competitive ratio of classical
full-space facility location with nonuniform opening costs is at least $3$.
\end{corollary}

\begin{proof}
Use the same metric
\[
 \mathcal M:=\mathcal X\cup\mathcal Y
\]
from the candidate-site construction, but now allow every point of
$\mathcal M$ to be opened as a facility.  Set
\[
 f_{y_I}:=1
 \qquad (y_I\in\mathcal Y),
 \qquad
 f_{x_i}:=1+\frac N3
 \qquad (x_i\in\mathcal X).
\]
For every realization of the $N$ demand points, a suitable subset facility
$y_I$ still gives cost $1+N/3$.  Conversely, a solution that opens a demand
location already pays at least $1+N/3$ in opening cost, while a solution using
only subset facilities pays at least one unit of opening cost and at least
$1/3$ for each assignment.  Hence
\[
 \OPT=1+\frac N3
\]
for every realization.

Fix an arbitrary online algorithm $\mathcal A$ for the full-space instance.
For every $i\in[m]$, fix an $N$-element set
$I(i)\subseteq[m]$ containing $i$, and define a replacement map
\[
 \varphi(y_I):=y_I,
 \qquad
 \varphi(x_i):=y_{I(i)}.
\]
We construct an online algorithm $\mathcal A'$ for the restricted
candidate-site instance.

Algorithm $\mathcal A'$ maintains an internal simulation of the virtual
full-space state of $\mathcal A$, in addition to its own actual restricted
open set.  At every round, it feeds the arriving request to the simulation of
$\mathcal A$, using the same random bits when $\mathcal A$ is randomized.
Whenever the simulated algorithm opens a facility $z$, algorithm
$\mathcal A'$ opens $\varphi(z)$ in the same round, unless that replacement
facility is already open.  Thus the simulation depends only on the revealed
request prefix and is itself online.

We compare the two costs pathwise.  For every demand location $x_j$ and every
$i\in[m]$,
\begin{equation}
 d\bigl(x_j,\varphi(x_i)\bigr)
 =
 d\bigl(x_j,y_{I(i)}\bigr)
 \le
 d(x_j,x_i)+\frac13.
 \label{eq:replacement-distance}
\end{equation}
Indeed, equality holds when $j=i$ and when $j\notin I(i)$, while the remaining
case gives a smaller left-hand side.  Moreover, if $z=y_I$ is an original
subset facility, then
\[
 d\bigl(x_j,\varphi(z)\bigr)=d(x_j,z).
\]

Consider any round and let $z$ be the virtual facility to which
$\mathcal A$ assigns the current request.  The replacement $\varphi(z)$ has
already been opened by $\mathcal A'$ no later than that round.  Therefore
$\mathcal A'$ can assign the request to $\varphi(z)$, and assigning it to its
nearest actual open facility can only be cheaper.  Its connection cost in
that round is consequently at most the connection cost of $\mathcal A$ plus
$1/3$.  If $\mathcal A$ never opens a demand location, every facility it opens
is an original subset facility and is reproduced unchanged by
$\mathcal A'$, so no connection-cost increase occurs.
Otherwise, summing over the $N$ rounds shows that the total connection cost
of $\mathcal A'$ exceeds that of $\mathcal A$ by at most $N/3$.

Let $K_X$ be the number of distinct demand locations opened by
$\mathcal A$.  Each such facility has cost $1+N/3$ and is replaced by a
facility of cost one.  Duplicate replacements are opened only once, so the
total opening cost decreases by at least $K_X\frac N3$.
Openings of original subset facilities are reproduced at the same cost, or at
smaller cost if the corresponding facility is already open.  Hence, when
$K_X\ge1$,
\[
 \operatorname{cost}(\mathcal A')
 \le
 \operatorname{cost}(\mathcal A)
 +\frac N3-K_X\frac N3
 \le
 \operatorname{cost}(\mathcal A).
\]
When $K_X=0$, neither the opening cost nor the connection cost increases, so
the same inequality holds.  The comparison is pathwise and therefore also
holds after taking expectations over the random order and the internal
randomness of $\mathcal A$.

Apply \cref{thm:lower-three} to $\mathcal A'$.  It provides a fixed
random-order demand multiset on which
\[
 \E[\operatorname{cost}(\mathcal A')]
 \ge
 N.
\]
On the same multiset,
$\operatorname{cost}(\mathcal A)\ge
\operatorname{cost}(\mathcal A')$ pathwise and
$\OPT=1+N/3$.  The claimed ratio follows.
\end{proof}

\begin{corollary}[Uniform versus nonuniform opening costs]
\label{cor:separation}
Let $\operatorname{CR}^{\mathrm{full}}_{\mathrm{unif}}$ and
$\operatorname{CR}^{\mathrm{full}}_{\mathrm{nonunif}}$ denote the
optimal random-order competitive ratios, in the known-horizon model,
of classical full-space metric facility location with uniform and
nonuniform opening costs, respectively.  Then
\[
 \operatorname{CR}^{\mathrm{full}}_{\mathrm{unif}}
 <2.42<3
 \le
 \operatorname{CR}^{\mathrm{full}}_{\mathrm{nonunif}}.
\]
In particular, the two models have strictly different optimal competitive
ratios.
\end{corollary}

\begin{proof}
Huang and Jiang~\cite{HuangJiang2026} give a deterministic algorithm with
competitive ratio below $2.42$ for the classical full-space uniform-cost
model.  The nonuniform lower bound follows from
\cref{cor:fullspace-lower} by letting $N\to\infty$.
\end{proof}

\section{Discussion}
\label{sec:discussion}
We obtain a $4.2674$ upper bound for random-order online facility location with
an arbitrary finite candidate set and nonuniform opening costs.  Our explicit
candidate-site notation should not be interpreted as a new restriction absent
from the classical literature.  Meyerson's nonuniform node-cost formulation
implicitly encodes forbidden sites through infinite opening costs.  Our result
improves his factor $33$ for that general formulation, and it also improves the
factor $8$ proved by Li et al. for the explicitly stated uniform-cost
candidate-site problem.  Penalized-Distance RankCut compares all candidates
through the single objective $d(x,y)+\lambda_t f_y$ and therefore avoids
cost discretization and separate opening decisions at different scales.  For
unit costs, \cref{prop:unit-specialization} shows that this rule reduces
exactly to the shifted-rank Candidate Marginal DistCut cutoff.
Appendix~\ref{app:uniform-candidates} sharpens the analysis in that setting
by using the closely related zero-start normalization, obtaining the
supplementary $3.2805$ guarantee.  Its class-optimality proof adapts the
TimeDist framework of Huang and Jiang and applies only to rules based on rank
and marginal improvement, not to unrestricted candidate-aware algorithms.
The candidate-specific change is the residual connection distance left after
opening a nearest feasible site.  We therefore treat the sharper unit-cost
bound as a supplementary specialization rather than a separate main
contribution.  Taking $\F=M$ gives the classical finite full-space model
directly. 

The upper-bound algorithm uses the known horizon explicitly through the
normalized rank $q_t=t/n$.  Extending the same guarantee to an unknown horizon
remains open: a direct doubling argument does not preserve the global rank
clock on which the analysis depends.

Our lower bound has two complementary interpretations.  The first construction
shows a factor of three with uniform costs on an explicit restricted candidate
set.  The replacement argument in \cref{cor:fullspace-lower} then removes the
facility restriction altogether: the same lower bound persists when every
metric point is feasible and all opening costs are finite, once the costs may
be nonuniform.  Thus the factor-three obstruction is not an artifact of the
candidate-site presentation or of using infinite costs to forbid demand
locations.  Combined with the sub-$2.42$ uniform-cost algorithm of Huang and
Jiang, it yields a strict separation between the optimal competitive ratios
of the classical full-space uniform and nonuniform models.

The radius coefficient $3+4\mu$ remains the main loss in the upper-bound
analysis.  The proof routes every later cluster point through the first one
and pays a separate high-cost excess at that first arrival.  Improving the
constant appears to require a tighter coupling between these two charges, or
an online mechanism that accumulates support for an expensive facility across
several arrivals rather than certifying it from a single request.

\section*{Acknowledgments and AI disclosure}

The authors designed the paper's overall organization and narrative.
ChatGPT 5.5 and 5.6 were used to prepare initial drafts within this framework
and to assist in shaping the statement and proof of
\cref{lem:envelope}.  All generated material was reviewed and revised by the
authors.  The authors' initial analysis used independent continuous
$\operatorname{Unif}[0,1]$ clocks and obtained the same competitive-ratio
bounds reported in this paper.  ChatGPT suggested derandomizing this
formulation via deterministic normalized arrival ranks, as well as jointly
bounding the base opening cost and the anchor excess in
\cref{lem:first-base}.  The authors take full responsibility for the paper's
correctness and presentation.
% \section*{Acknowledgments and AI disclosure}

% The authors designed the paper's overall organization and narrative.
% ChatGPT 5.5 and 5.6 were used to prepare initial drafts within this framework
% and to assist in shaping the statement and proof of
% \cref{lem:envelope}.  All generated material was reviewed and revised by the
% authors.  ChatGPT also assisted in identifying the construction underlying
% the $3-o(1)$ lower bound.  It further suggested derandomizing an initial
% formulation based on independent continuous $\operatorname{Unif}[0,1]$ clocks
% via deterministic normalized arrival ranks, and jointly bounding the base
% opening cost and the anchor excess in \cref{lem:first-base}.  The authors take
% full responsibility for the paper's correctness and presentation.

\clearpage
\appendix

\section{Unit-cost prescribed candidates: the zero-start rank cutoff}
\label{app:uniform-candidates}

The main algorithm is already deterministic.  Under unit opening costs,
\cref{prop:unit-specialization} shows that it becomes Candidate Marginal
DistCut with the shifted rank $q_t=t/n$.  For the sharper unit-cost analysis,
this appendix instead uses the zero-start rank
\[
 z_t:=\frac{t-1}{n},
\]
which is the normalization used by Huang and Jiang's full-space
$\mu$-DistCut framework~\cite{HuangJiang2026}.  The two rules differ only by
one rank step, but the distinction matters in the nonuniform model:
\cref{prop:zero-start-unbounded} shows that a zero opening-cost penalty at the
first request can select an arbitrarily expensive facility.  With unit costs,
all candidates have the same opening cost, so that obstruction disappears.

We keep the rank-cutoff steps shared with Huang and Jiang brief.  The points
specific to prescribed candidates are the residual connection distance left
after opening a nearest feasible site, the signed baseline cancellation in
the cluster analysis, and the private-candidate geometry in the class lower
bound.

For a request $x$, current open set $O$, and
\[
 D=d(x,O),
 \qquad
 a_x=d(x,\F),
\]
define
\[
 \Delta(x;D)
 :=
 D-\min\{D,a_x\}
 =
 (D-a_x)^+.
\]
Because $f_y=1$ for every candidate, the choice in
\eqref{eq:penalized-choice} minimizes $d(x,y)$.  Thus the shifted-rank opening
test of Algorithm~\ref{alg:main} is exactly
\[
 \Delta(x;D)\ge
 \min\left\{1,\frac{t}{\mu n}\right\}.
\]
The zero-start rule analyzed below replaces $t/n$ by $(t-1)/n$.  This choice
aligns exactly with Huang and Jiang and supports the tighter single-scale
certificate.

\subsection{A deterministic zero-start rule}

For $\mu\in(0,1]$, put
\[
 \theta_\mu(z):=\min\left\{1,\frac{z}{\mu}\right\}.
\]

\begin{algorithm}[H]
\caption{Zero-Start Candidate Marginal DistCut}\label{alg:uniform-candidate}
\begin{algorithmic}[1]
\State $O_0\gets\varnothing$.
\For{$t=1,\ldots,n$}
  \State Receive $x=v_t$ and set
         $D\gets d(x,O_{t-1})$ and $a\gets d(x,\F)$.
  \State $\Delta\gets D-\min\{D,a\}$ and
         $z_t\gets(t-1)/n$.
  \State $O_t\gets O_{t-1}$.
  \If{$\Delta\ge\theta_\mu(z_t)$}
    \State Choose a nearest candidate $y(x)\in\F$ to $x$.
    \State $O_t\gets O_t\cup\{y(x)\}$.
  \EndIf
  \State Connect $x$ to its nearest facility in $O_t$.
\EndFor
\end{algorithmic}
\end{algorithm}

When $O_{t-1}=\varnothing$, $\Delta=+\infty$, so the first request opens a
facility.  A fixed nearest-candidate tie-breaking rule makes the algorithm
deterministic.

\begin{remark}[Why the appendix retains the zero-start rank]
\label{rem:uniform-clock}
As explained in \cref{rem:continuous-clock}, the positive rank $t/n$ used by
the main algorithm and the zero-start rank $(t-1)/n$ used here are two
deterministic discretizations of a common continuous clock.  The former
prevents the unbounded first-step behavior of
\cref{prop:zero-start-unbounded}.  That obstruction disappears at unit costs,
so we retain the zero-start convention for its sharper first-point certificate
and its direct alignment with Huang and Jiang.

The one-rank difference cannot simply be absorbed by the present certificate.
Replacing $z_T$ by
$q_T=z_T+\frac1n$
increases the failed-cutoff bound by at most $1/(n\mu)$.  Since the resulting
first-request residual is multiplied by a coefficient of order $k$ in a
cluster of size $k$, this produces a possible additional term of order $\frac{k}{n\mu}$,
which need not vanish when $k=\Theta(n)$.  This observation concerns the
certificate and does not imply a pointwise performance ordering between the
two deterministic algorithms.
\end{remark}

Define
\[
 \Phi(\mu):=
 1+\frac{e^{-(1+\mu)}}{\mu}.
\]

\begin{theorem}[Deterministic zero-start unit-cost bound]
\label{thm:uniform-candidate-upper}
For every $\mu\in(0,1]$,
Algorithm~\ref{alg:uniform-candidate} satisfies
\[
 \E[\ALG]
 \le
 \max\{\Phi(\mu),\,3+2\mu\}\OPT.
\]
\end{theorem}

\paragraph{Rank facts.}
Fix an optimal cluster $C$ of $k$ demand copies.  Let $p$ be its first copy
in the random permutation and $T$ its rank.  By \cref{lem:clock}, $p$ is
uniform in $C$, and, conditional on the ordered prefix through $T$, the rank
of each fixed $u\in C\setminus\{p\}$ is uniform on
$\{T+1,\ldots,n\}$.  In particular, since
$z_\ell=(\ell-1)/n$, \eqref{eq:first-rank-tail} gives, for every feasible
$\ell$,
\begin{equation}
 \Prb(T\ge\ell)
 =
 \frac{\binom{n-\ell+1}{k}}{\binom nk}
 \le
 (1-z_\ell)^k.
 \label{eq:uniform-rank-tail}
\end{equation}

\paragraph{A cluster bound.}
Let $c\in\F$ serve $C$ in the fixed optimum, and write
\[
 r_u=d(u,c),
 \qquad
 R(C)=\sum_{u\in C}r_u,
 \qquad
 a_u=d(u,\F),
 \qquad
 A(C)=\sum_{u\in C}a_u.
\]
Since $c\in\F$,
\begin{equation}
 a_u\le r_u
 \quad\text{and}\quad
 A(C)\le R(C).
 \label{eq:uniform-a-radius}
\end{equation}

Let $D_p$ be the connection distance paid after serving $p$, let
$L_p\in\{0,1\}$ be its opening cost, and define
\[
 \eps_p:=(D_p-a_p)^+.
\]
If $p$ opens, then $D_p=a_p$ and $\eps_p=0$.  If it does not open, the failed
cutoff gives
\[
 D_p\le a_p+\frac{z_T}{\mu}.
\]
The same inequality is therefore valid in both cases.

\begin{lemma}[Later request]\label{lem:uniform-later}
Let $\mathcal K_p$ be the observed history through the service of $p$.  For every fixed demand copy $u\in C$, on the event $u\ne p$,
\[
 \E[\ALG(u)\mid\mathcal K_p]
 \le
 D_p+d(p,u)
 +\mu\bigl(a_p+d(p,u)-a_u\bigr).
\]
\end{lemma}

\begin{proof}
Set
\[
 B_u:=D_p+d(p,u),
 \qquad
 \delta_u:=(B_u-a_u)^+.
\]
The facility serving $p$ remains open, so the actual pre-service connection
distance $\widehat D_u$ of $u$ is at most $B_u$.
If $\delta_u=0$, then
\[
 \bigl(\widehat D_u-a_u\bigr)^+
 \le
 \bigl(B_u-a_u\bigr)^+
 =0.
\]
Since $u$ arrives after $p$, its rank satisfies $J>T\ge1$, and hence
$z_J>0$ and $\theta_\mu(z_J)>0$.  Therefore $u$ cannot trigger an
opening, and its cost is at most $\widehat D_u\le B_u$.

If $\delta_u\ge1$, then an opening costs at most
\[
 1+a_u\le B_u,
\]
while not opening also costs at most $B_u$.

It remains to consider the case $0<\delta_u<1$. An opening has cost
at most
\[
 1+a_u=B_u+(1-\delta_u),
\]
and can occur only at a future rank $J$ satisfying
\[
 z_J\le\mu\delta_u.
\]

Conditional on $\mathcal K_p$, $J$ is uniform on
$\{T+1,\ldots,n\}$.  The number of these ranks satisfying the last inequality
is either zero or at most
\[
 n\mu\delta_u-T+1
 =
 n(\mu\delta_u-z_T).
\]
If it is nonzero, then $T/n\le\mu\delta_u\le\delta_u$, where the last
inequality uses $\mu\le1$.  Hence the expected premium above $B_u$ is at most
\begin{align*}
 (1-\delta_u)
 \frac{n(\mu\delta_u-z_T)}{n-T}
 &\le
 \mu\delta_u-z_T.
\end{align*}
If there is no such rank, the premium is zero.

The first-point certificate gives
\[
 B_u-a_u
 \le
 a_p+d(p,u)-a_u+\frac{z_T}{\mu}.
\]
The map $x\mapsto d(x,\F)$ is $1$-Lipschitz, so
$a_u\le a_p+d(p,u)$.  Consequently,
\[
 \mu\delta_u-z_T
 \le
 \mu\bigl(a_p+d(p,u)-a_u\bigr).
\]
Adding this premium to $B_u$ proves the lemma.
\end{proof}

\begin{lemma}[First request]\label{lem:uniform-first}
For every cluster $C$ of size $k$,
\[
 \E[L_p+k\eps_p]\le\Phi(\mu).
\]
\end{lemma}

\begin{proof}
Fix the identity of $p$ and the relative order of all requests outside
$C$.  By the same rank-set representation used in
\cref{lem:clock}, the set of ranks occupied by $C$ is independent of
both the identity of $p$ and the relative order of the copies outside
$C$.  Hence, under the present conditioning, the occupied rank set
remains a uniformly random $k$-subset of $[n]$, and the tail bound
\eqref{eq:uniform-rank-tail} remains valid conditionally.

Couple the feasible values
$t\in\{1,\ldots,n-k+1\}$ of the first-cluster rank as follows.  When
$p$ is placed at rank $t$, the first $t-1$ requests are the first
$t-1$ requests in the fixed outside-$C$ order.  Let
\[
 \widehat D_p(t):=d(p,O_{t-1})
\]
be the pre-service connection distance of $p$ in this coupled run. 

Define
\[
 \widehat\eps_p(t)
 :=
 \bigl(\widehat D_p(t)-a_p\bigr)^+.
\]
Compare the coupled runs in which $p$ is placed at ranks $t$ and $t+1$.
The requests in ranks $1,\ldots,t-1$ are the same outside-$C$ requests in
the same ranks in both runs.  Since the algorithm is deterministic, their
histories and decisions are identical.  In the latter run, one additional
outside-$C$ request is processed at rank $t$ before $p$ arrives.  Facilities
never close, so the open set available when $p$ is served can only grow.
Consequently, $\widehat D_p(t)$ and $\widehat\eps_p(t)$ are nonincreasing in
$t$.  On the other hand,
$\theta_\mu(z_t)$ is nondecreasing in $t$.  Hence the ranks satisfying
\[
 \widehat\eps_p(t)\ge\theta_\mu(z_t)
\]
form an initial interval.

If every feasible rank triggers, then $p$ opens a nearest candidate and
$L_p+k\eps_p=1$.  Otherwise, let $\ell$ be the first non-triggering
rank and put $s=z_\ell$.  If $T<\ell$, then $p$ triggers and
$L_p+k\eps_p\le1$.  If $T\ge\ell$, then $p$ does not trigger and its
post-service distance equals its pre-service distance.  Therefore,
\[
 \eps_p
 =
 \bigl(\widehat D_p(T)-a_p\bigr)^+
 \le
 \bigl(\widehat D_p(\ell)-a_p\bigr)^+
 <
 \theta_\mu(z_\ell)
 =
 \theta_\mu(s).
\]
By the conditional form of \eqref{eq:uniform-rank-tail}, the expected
charge under the present conditioning is at most
\[
 1+\bigl(k\theta_\mu(s)-1\bigr)(1-s)^k
\]
whenever the coefficient is positive, and at most one otherwise.

For $s\le\mu$, put $x=ks$ and use $(1-s)^k\le e^{-x}$:
\[
 1+\left(\frac{x}{\mu}-1\right)e^{-x}
 \le
 1+\frac{e^{-(1+\mu)}}{\mu},
\]
because the second term is maximized at $x=1+\mu$.  For $s\ge\mu$,
\[
 1+(k-1)(1-s)^k
 \le
 1+(k-1)e^{-k\mu}
 \le
 1+\frac{e^{-(1+\mu)}}{\mu}.
\]
The bound is independent of the conditioning.
\end{proof}

\begin{proof}[Proof of \cref{thm:uniform-candidate-upper}]
For the random first request $p$, let
\[
 Y_C:=\sum_{u\in C}\1\{u\ne p\}\ALG(u).
\]
The identity of $p$ is $\mathcal K_p$-measurable.  Conditional linearity and
\cref{lem:uniform-later} give
\[
 \E[Y_C\mid\mathcal K_p]
 \le
 (k-1)D_p
 +(1+\mu)\sum_{u\ne p}d(p,u)
 +\mu\bigl(ka_p-A(C)\bigr).
\]
Since $L_p$ and $D_p$ are also $\mathcal K_p$-measurable, the tower property
gives
\begin{align}
 \E[\ALG(C)]
 \le
 \E\!\left[
  L_p+kD_p
  +(1+\mu)\sum_{u\ne p}d(p,u)
  +\mu\bigl(ka_p-A(C)\bigr)
 \right].
 \label{eq:uniform-cluster-tower}
\end{align}
No independence among the later costs is used.

Using $D_p\le a_p+\eps_p$ and \cref{lem:uniform-first},
\begin{align*}
 \E[\ALG(C)]
 \le{}&
 \Phi(\mu)
 +\E[ka_p]
 +(1+\mu)\E\left[\sum_{u\ne p}d(p,u)\right]\\
 &+\mu\,\E[ka_p-A(C)].
\end{align*}
Since $p$ is uniform in $C$,
\[
 \E[ka_p]=A(C),
 \qquad
 \E[ka_p-A(C)]=0.
\]
Furthermore,
\[
 \sum_{u\ne p}d(p,u)
 \le
 R(C)+(k-2)r_p,
\]
and $\E[r_p]=R(C)/k$.  With \eqref{eq:uniform-a-radius},
\[
 \E[\ALG(C)]
 \le
 \Phi(\mu)
 +A(C)
 +2(1+\mu)\left(1-\frac1k\right)R(C)
 \le
 \Phi(\mu)+(3+2\mu)R(C).
\]
Summing over the clusters of an optimum proves the theorem.
\end{proof}

Let $\mu^*$ denote the unique solution of
\[
 \Phi(\mu)=3+2\mu,
\]
and define
\[
 C^*:=\Phi(\mu^*)=3+2\mu^*.
\]
Numerically,
\[
 \mu^*\approx0.140215,
 \qquad
 C^*\approx3.28043.
\]
For example, $\mu=0.140215$ gives a ratio below $3.2805$ for the deterministic
zero-start rule.

\subsection{Comparison with the TimeDist class}

Huang and Jiang define TimeDist rules whose opening probability is an
arbitrary function of the rank and current connection distance.  Their class
lower bound extracts the average opening propensity at one distance scale and
uses a sparse star when that propensity is large and repeated dense locations
when it is small~\cite{HuangJiang2026}.  The analogous scalar here is the
marginal improvement.

\begin{definition}[Candidate Marginal TimeDist]
A Candidate Marginal TimeDist rule is specified by a family of
instance-independent functions
\[
 \mathcal G=(g_n)_{n\ge1},
 \qquad
 g_n:\{2,\ldots,n\}\times\R_{\ge0}\to[0,1],
\]
where the first set is empty when $n=1$.  At rank $1$, the rule opens a fixed
nearest candidate.  At every rank $t\ge2$, let $\mathcal H_t$ denote the
complete revealed history immediately before the opening decision, including
the request prefix, the current open set, and all previous actions and random
outcomes.  Conditional on $\mathcal H_t$ and the current request, the rule
opens a fixed nearest candidate with probability exactly
\[
 g_n(t,\Delta_t),
 \qquad
 \Delta_t=
 \bigl(d(v_t,O_{t-1})-d(v_t,\F)\bigr)^+.
\]
The nearest-candidate tie-breaking is fixed in advance.
\end{definition}

This definition deliberately restricts both the information used by the
opening probability and the facility selected after an opening decision.  It
does not include rules that separately inspect the current distance and the
nearest-candidate distance, use candidate identities, or open a non-nearest
candidate.

\begin{theorem}[Optimality within Candidate Marginal TimeDist]
\label{thm:uniform-class-tight}
For every Candidate Marginal TimeDist rule
$\mathcal G=(g_n)_{n\ge1}$,
\[
 \limsup_{n\to\infty}
 \sup_{\substack{I:\,|I|=n}}
 \frac{\E[\ALG_{\mathcal G}(I)]}{\OPT(I)}
 \ge C^*,
\]
where $\ALG_{\mathcal G}(I)$ denotes the cost incurred by the rule
$\mathcal G$ on instance $I$, and the expectation is over the uniformly
random arrival order and the rule's internal randomness.  The supremum ranges
over unit-cost prescribed-candidate instances with $n$ demand copies.
Consequently, Candidate Marginal DistCut with $\mu=\mu^*$ is asymptotically
optimal in this class; specifically, the
zero-start rule in Algorithm~\ref{alg:uniform-candidate} attains the matching
upper bound.

This optimality statement is confined to the information-restricted class
in the preceding definition; it makes no claim about unrestricted
candidate-aware algorithms.
\end{theorem}

\begin{proof}
We adapt the sparse-versus-dense dichotomy of Huang and
Jiang~\cite{HuangJiang2026}.  The repeated-location branch is the same after
replacing their distance signal by the marginal signal.  The sparse branch
requires a private candidate at each request and is included explicitly.

For a large integer $K$, set
\[
 n:=K^3,
 \qquad
 L^*:=1+\frac1{\mu^*},
 \qquad
 d:=\frac{L^*}{K}.
\]
Let
\[
 p_1:=1,
 \qquad
 p_t:=g_n(t,d)\quad(t\ge2),
 \qquad
 \overline p:=\frac1n\sum_{t=1}^np_t,
 \qquad
 \eta_K:=\frac{\overline p}{d}.
\]
The two hard instances below have the same horizon and marginal scale $d$.

\smallskip
\noindent\emph{Large propensity: $\eta_K\ge\mu^*$.}
Put $r=d/2$ and $\eps=d^2$.  Since $K\to\infty$, we may assume
\[
 K>2L^*,
\]
so that $d<1/2$ and
\[
 r-\eps
 =
 \frac d2-d^2
 =
 d\left(\frac12-d\right)
 >0.
\]
Thus all edges in the following tree have positive length.

Consider a tree with a central candidate $c$.  For every $i\in[n]$, attach
a request location $u_i$ to $c$ by an edge of length $r$, and attach a
private candidate $y_i$ to $u_i$ by an edge of length $r-\eps$. There is one request at every $u_i$.

After the first request, a fresh $u_i$ is at distance $3r-\eps$ from every
open private candidate, while its unique nearest candidate $y_i$ is at
distance $r-\eps$.  Its marginal improvement is therefore $2r=d$, and it
opens with probability $p_t$ at rank $t$.  Opening costs
$1+r-\eps$ in that round; not opening costs $3r-\eps$.  Hence
\[
 \E[\ALG]
 =
 n\left(
  \frac{3d}{2}-d^2+(1-d)d\eta_K
 \right).
\]
Opening $c$ alone gives $\OPT\le1+nd/2$, so
\[
 \frac{\E[\ALG]}{\OPT}
 \ge
 \frac{3-2d+2(1-d)\eta_K}{1+2/(nd)}
 \ge
 C^*-o_K(1).
\]
The last inequality uses $\eta_K\ge\mu^*$ directly; this avoids requiring
any upper bound on $\eta_K$.

\smallskip
\noindent\emph{Small propensity: $\eta_K\le\mu^*$.}
Take $M=K^2$ candidate locations at pairwise distance $d$, and place $K$
copies at each location.  Thus $MK=n$.  Until a location opens, each of its
requests has marginal improvement $d$.

Fix one location and let $S$ be the uniformly random $K$-subset of ranks
occupied by its copies.  If $1\in S$, then this location opens at the first
request, and the probability that it never opens is zero; this agrees with the
product below because $p_1=1$.

Suppose now that $1\notin S$.  The request at rank one opens a different
candidate location.  On the event that the fixed location is still unopened
immediately before a rank $t\in S$, its nearest candidate is at distance zero
and every open candidate at another location is at distance $d$.  Its marginal
improvement is therefore exactly $d$.  By the definition of Candidate
Marginal TimeDist, conditional on the complete revealed history, the location
remains unopened in that round with probability $1-p_t$.  Iterating these
conditional probabilities over the ranks in $S$ and applying the chain rule
gives
\[
 \Prb(\text{the location never opens}\mid S)
 =
 \prod_{t\in S}(1-p_t).
\]
No independence of the random choices across rounds is required.

Denote the unconditional probability in the last display by $P_0$.  If
$T_1,\ldots,T_K$ are independent uniform ranks and $\mathcal D$ is the event
that they are distinct, then, conditional on $\mathcal D$, their unordered
set is a uniform $K$-subset.  Since all factors lie in $[0,1]$,
\[
 (1-\overline p)^K
 \le
 P_0+\Prb(\mathcal D^c),
\]
and
\[
 \Prb(\mathcal D^c)
 \le
 \frac{K(K-1)}{2n}
 =
 O(1/K).
\]
Because $\overline p=d\eta_K\le d\mu^*$ and $Kd=L^*$,
\[
 P_0
 \ge
 (1-d\mu^*)^K-o(1)
 =
 e^{-\mu^*L^*}-o(1).
\]

If the location ever opens, its total contribution is at least one; if it
never opens, its $K$ copies cost $Kd=L^*$.  Opening every location gives
$\OPT\le M$.  Consequently,
\begin{align*}
 \frac{\E[\ALG]}{\OPT}
 &\ge
 1+(L^*-1)P_0\\
 &\ge
 1+(L^*-1)e^{-\mu^*L^*}-o(1)\\
 &=
 \Phi(\mu^*)-o(1)
 =
 C^*-o(1).
\end{align*}

For every sufficiently large $K$, one of the two cases applies at horizon
$n=K^3$.  Hence the displayed lower bound holds along an infinite subsequence
of horizons, which proves the claimed $\limsup$ bound.  The matching upper
bound is \cref{thm:uniform-candidate-upper}.
\end{proof}

\paragraph{Relation to the full-space result.}
The zero-start cutoff, first-rank argument, and sparse-versus-dense
dichotomy above are deliberately parallel to Huang and Jiang.  In contrast,
the main nonuniform-cost algorithm uses the shifted rank $t/n$; the shift is
invisible asymptotically at a fixed scale but prevents the unbounded first-step
failure in \cref{prop:zero-start-unbounded}.

For full-space DistCut, opening at the request leaves zero connection
distance, and the sparse branch gives $2+2\mu$.  Here an opening leaves the
residual distance $d(x,\F)$; in the upper bound this contributes
$A(C)\le R(C)$, and in the lower bound the private-candidate star makes this
extra charge unavoidable.  This changes the increasing term to
$3+2\mu$.  The proof above should therefore be read as the prescribed-site
counterpart of the TimeDist argument, with the private-candidate geometry as
the substantive modification.  It does not identify the optimal competitive
ratio of unrestricted candidate-aware algorithms, which remains between
$3$ and $3.2805$.

\section{Why one-center lower bounds stop at three}\label{app:one-center}

For completeness, the value three is tight for the following broad
single-facility benchmark. This explains why improving the lower bound requires
instances whose offline optimum genuinely uses a hidden multi-facility
partition.

\begin{proposition}[First-demand-point representative]
\label{prop:one-center-three}
Let $U$ be a multiset of $k\ge1$ demand points in an arbitrary finite
candidate-facility metric with arbitrary nonnegative opening costs. Define
\[
 \OPT_1(U)
 =
 \min_{c\in\F}
 \left\{
  f_c+\sum_{u\in U}d(u,c)
 \right\}.
\]
If $k$ is known and the demand points arrive in uniformly random order, there is an
online algorithm with expected cost at most $3\OPT_1(U)$.
\end{proposition}

\begin{proof}
Let $p$ be the first demand point and choose
\[
 y_p\in\arg\min_{y\in\F}
 \left\{
  f_y+kd(p,y)
 \right\}.
\]
Open $y_p$ and connect every demand point to it. Let $c$ attain $\OPT_1(U)$, and
write
\[
 f:=f_c,
 \qquad
 r_u:=d(u,c),
 \qquad
 R:=\sum_{u\in U}r_u.
\]
By the choice of $y_p$,
\[
 f_{y_p}+kd(p,y_p)\le f+kr_p.
\]
For every $u\in U$, the triangle inequality gives
\[
 d(u,y_p)\le r_u+r_p+d(p,y_p).
\]
Therefore,
\begin{align*}
 \ALG
 &=
 f_{y_p}+\sum_{u\in U}d(u,y_p)\\
 &\le
 R+kr_p+f_{y_p}+kd(p,y_p)\\
 &\le
 f+R+2kr_p.
\end{align*}
The first demand point is uniform over the $k$ demand-point copies, so
\[
 \E[r_p]=\frac{R}{k}.
\]
Hence
\[
 \E[\ALG]
 \le
 f+3R
 \le
 3(f+R)
 =
 3\OPT_1(U).
\]
\end{proof}

\begin{remark}
The lower-bound construction of \cref{thm:lower-three} has an optimal solution
using one facility and asymptotically attains the factor three. Thus no
one-facility hidden-center construction can yield a lower bound strictly larger
than three when the cluster size is known.
\end{remark}


\begin{thebibliography}{FGGPT25}

\bibitem[AKL21]{AlbersKhanLadewig2021}
Susanne Albers, Arindam Khan, and Leon Ladewig.
\newblock Improved online algorithms for knapsack and GAP in the random order
model.
\newblock \emph{Algorithmica}, 83(6):1750--1785, 2021.

\bibitem[Alm+21]{AlmanzaEtAl2021}
Matteo Almanza, Flavio Chierichetti, Silvio Lattanzi, Alessandro Panconesi,
and Giuseppe Re.
\newblock Online facility location with multiple advice.
\newblock In \emph{Advances in Neural Information Processing Systems 34
(NeurIPS)}, 2021.

\bibitem[ABUV04]{AnagnostopoulosEtAl2004}
Aris Anagnostopoulos, Russell Bent, Eli Upfal, and Pascal Van Hentenryck.
\newblock A simple and deterministic competitive algorithm for online facility
location.
\newblock \emph{Information and Computation}, 194(2):175--202, 2004.

\bibitem[AFGS22]{ArgueEtAl2022}
C.~J. Argue, Alan Frieze, Anupam Gupta, and Christopher Seiler.
\newblock Learning from a sample in online algorithms.
\newblock In \emph{Advances in Neural Information Processing Systems 35
(NeurIPS)}, pages 13852--13863, 2022.

\bibitem[APT22]{AzarPanigrahiTouitou2022}
Yossi Azar, Debmalya Panigrahi, and Noam Touitou.
\newblock Online graph algorithms with predictions.
\newblock In \emph{Proceedings of the 2022 Annual ACM-SIAM Symposium on
Discrete Algorithms (SODA)}, pages 35--66, 2022.

\bibitem[BGSZ20]{BradacEtAl2020}
Domagoj Bradac, Anupam Gupta, Sahil Singla, and Goran Zuzic.
\newblock Robust algorithms for the secretary problem.
\newblock In \emph{11th Innovations in Theoretical Computer Science
Conference (ITCS)}, 2020.

\bibitem[COP03]{CharikarEtAl2003}
Moses Charikar, Liadan O'Callaghan, and Rina Panigrahy.
\newblock Better streaming algorithms for clustering problems.
\newblock In \emph{Proceedings of the 35th Annual ACM Symposium on Theory of
Computing (STOC)}, pages 30--39, 2003.

\bibitem[CCMS18]{CyganEtAl2018}
Marek Cygan, Artur Czumaj, Marcin Mucha, and Piotr Sankowski.
\newblock Online facility location with deletions.
\newblock In \emph{26th Annual European Symposium on Algorithms (ESA)}, pages
21:1--21:15, 2018.

\bibitem[Fot07]{Fotakis2007}
Dimitris Fotakis.
\newblock A primal--dual algorithm for online non-uniform facility location.
\newblock \emph{Journal of Discrete Algorithms}, 5(1):141--148, 2007.

\bibitem[Fot08]{Fotakis2008}
Dimitris Fotakis.
\newblock On the competitive ratio for online facility location.
\newblock \emph{Algorithmica}, 50(1):1--57, 2008.

\bibitem[Fot11]{FotakisSurvey2011}
Dimitris Fotakis.
\newblock Online and incremental algorithms for facility location.
\newblock \emph{ACM SIGACT News}, 42(1):97--131, 2011.

\bibitem[FGGPT25]{FotakisEtAl2025}
Dimitris Fotakis, Evangelia Gergatsouli, Themistoklis Gouleakis, Nikolas
Patris, and Thanos Tolias.
\newblock Improved bounds for online facility location with predictions.
\newblock In \emph{Proceedings of the AAAI Conference on Artificial
Intelligence}, 39(25):26973--26981, 2025.

\bibitem[GMM+03]{GuhaEtAl2003}
Sudipto Guha, Adam Meyerson, Nina Mishra, Rajeev Motwani, and Liadan
O'Callaghan.
\newblock Clustering data streams: Theory and practice.
\newblock \emph{IEEE Transactions on Knowledge and Data Engineering},
15(3):515--528, 2003.

\bibitem[GKLX20]{GuoEtAl2020}
Xiangyu Guo, Janardhan Kulkarni, Shi Li, and Jiayi Xian.
\newblock On the facility location problem in online and dynamic models.
\newblock In \emph{Approximation, Randomization, and Combinatorial
Optimization. Algorithms and Techniques (APPROX/RANDOM)}, volume 176 of
\emph{LIPIcs}, pages 42:1--42:23, 2020.

\bibitem[GKL22]{GuptaKehneLevin2022}
Anupam Gupta, Gregory Kehne, and Roie Levin.
\newblock Random order online set cover is as easy as offline.
\newblock In \emph{2021 IEEE 62nd Annual Symposium on Foundations of Computer
Science (FOCS)}, pages 1253--1264, 2022.

\bibitem[GS21]{GuptaSingla2021}
Anupam Gupta and Sahil Singla.
\newblock Random-order models.
\newblock In Tim Roughgarden, editor, \emph{Beyond the Worst-Case Analysis of
Algorithms}, pages 234--258. Cambridge University Press, 2021.

\bibitem[HJ26]{HuangJiang2026}
Yichen Huang and Shaofeng H.-C. Jiang.
\newblock The power of arrival times in random-order online facility location.
\newblock arXiv:2607.10564, 2026.

\bibitem[JLL+22]{JiangEtAl2022}
Shaofeng H.-C. Jiang, Erzhi Liu, You Lyu, Zhihao Gavin Tang, and Yubo Zhang.
\newblock Online facility location with predictions.
\newblock In \emph{International Conference on Learning Representations
(ICLR)}, 2022.

\bibitem[KNR20]{KaplanNaoriRaz2020}
Haim Kaplan, David Naori, and Danny Raz.
\newblock Competitive analysis with a sample and the secretary problem.
\newblock In \emph{Proceedings of the Thirty-First Annual ACM-SIAM Symposium
on Discrete Algorithms (SODA)}, pages 2082--2095, 2020.

\bibitem[KNR22]{KaplanNaoriRaz2022}
Haim Kaplan, David Naori, and Danny Raz.
\newblock Online weighted matching with a sample.
\newblock In \emph{Proceedings of the 2022 Annual ACM-SIAM Symposium on
Discrete Algorithms (SODA)}, pages 1247--1272, 2022.

\bibitem[KNR23]{KaplanNaoriRaz2023}
Haim Kaplan, David Naori, and Danny Raz.
\newblock Almost tight bounds for online facility location in the random-order
model.
\newblock In \emph{Proceedings of the 2023 Annual ACM-SIAM Symposium on
Discrete Algorithms (SODA)}, pages 1523--1544, 2023.

\bibitem[KRTV18]{KesselheimEtAl2018}
Thomas Kesselheim, Klaus Radke, Andreas T{\"o}nnis, and Berthold V{\"o}cking.
\newblock Primal beats dual on online packing LPs in the random-order model.
\newblock \emph{SIAM Journal on Computing}, 47(5):1939--1964, 2018.

\bibitem[KKN15]{KesselheimKleinbergNiazadeh2015}
Thomas Kesselheim, Robert Kleinberg, and Rad Niazadeh.
\newblock Secretary problems with non-uniform arrival order.
\newblock In \emph{Proceedings of the 47th Annual ACM Symposium on Theory of
Computing (STOC)}, pages 879--888, 2015.

\bibitem[KM20]{KesselheimMolinaro2020}
Thomas Kesselheim and Marco Molinaro.
\newblock Knapsack secretary with bursty adversary.
\newblock In \emph{47th International Colloquium on Automata, Languages, and
Programming (ICALP)}, pages 72:1--72:15, 2020.

\bibitem[Lan18]{Lang2018}
Harry Lang.
\newblock Online facility location against a $t$-bounded adversary.
\newblock In \emph{Proceedings of the Twenty-Ninth Annual ACM-SIAM Symposium
on Discrete Algorithms (SODA)}, pages 1002--1014, 2018.

\bibitem[LMWX26]{LiMiaoWuXu2026}
Mengzhen Li, Runjie Miao, Chenchen Wu, and Dachuan Xu.
\newblock On competitive ratio for online uniform facility location problem in
random-order model.
\newblock \emph{Theoretical Computer Science}, 1072:115878, 2026.

\bibitem[Mey01]{Meyerson2001}
Adam Meyerson.
\newblock Online facility location.
\newblock In \emph{Proceedings of the 42nd IEEE Symposium on Foundations of
Computer Science (FOCS)}, pages 426--431, 2001.

\bibitem[MOGZ12]{MirrokniEtAl2012}
Vahab S. Mirrokni, Shayan Oveis Gharan, and Morteza Zadimoghaddam.
\newblock Simultaneous approximations for adversarial and stochastic online
budgeted allocation.
\newblock In \emph{Proceedings of the Twenty-Third Annual ACM-SIAM Symposium
on Discrete Algorithms (SODA)}, pages 1690--1701, 2012.

\bibitem[Mol17]{Molinaro2017}
Marco Molinaro.
\newblock Online and random-order load balancing simultaneously.
\newblock In \emph{Proceedings of the Twenty-Eighth Annual ACM-SIAM Symposium
on Discrete Algorithms (SODA)}, pages 1638--1650, 2017.

\bibitem[NW13]{NagarajanWilliamson2013}
Chandrashekhar Nagarajan and David P. Williamson.
\newblock Offline and online facility leasing.
\newblock \emph{Discrete Optimization}, 10(4):361--370, 2013.

\end{thebibliography}
\end{document}